\documentclass[11pt]{article} 

\usepackage{geometry}
\usepackage{microtype}
\usepackage{graphicx}
\usepackage[usenames, dvipsnames]{xcolor}

\usepackage{appendix}
\usepackage{amsmath}
\usepackage{amsthm}
\usepackage{url}

\usepackage{amsfonts} 
\usepackage{tabularx}
\usepackage{tikz-cd}
\usepackage{algorithmicx}
\usepackage{accents}

\theoremstyle{definition}
\newtheorem{theorem}{Theorem}[section]

\newtheorem{lemma}[theorem]{Lemma}
\newtheorem{proposition}[theorem]{Proposition}
\newtheorem{corollary}[theorem]{Corollary}

\newtheorem{definition}[theorem]{Definition}

\theoremstyle{remark}
\newtheorem{remark}[theorem]{Remark}

\newcommand {\mm}[1] {\ifmmode{#1}\else{\mbox{\(#1\)}}\fi}

\newcommand{\Mspace}        {\mm{{\mathbb M}}}
\newcommand{\Rspace}        {\mm{{\mathbb R}}}
\newcommand{\Sspace}        {\mm{{\mathbb S}}}

\newcommand{\Xspace}        {\mm{{X}}}
\newcommand{\Yspace}        {\mm{{Y}}}

\newcommand{\Pset}        {\mm{\mathcal{P}}}

\newcommand{\zr}			{\mm{{\mathbb R}}}

\newcommand{\Scal}        {\mm{\mathcal S}}
\newcommand{\Vcal}        {\mm{\mathcal V}}

\newcommand{\cone}        {\mm{\mathrm{cone}}}

\newcommand{\bdr}        {{\partial}}
\newcommand{\dime}[1]       {\mm{\rm dim\,}{#1}}
\newcommand{\ep}            {\mm{\varepsilon}}

\newcommand{\para}[1]{\vspace{2mm}\noindent{\textbf{#1}}}
\newcommand{\denselist}{\vspace{-5pt} \itemsep -2pt\parsep=-1pt\partopsep -2pt}

\title{Discrete Stratified Morse Theory: Algorithms and A User's Guide}
\author{Kevin Knudson\\ University of Florida \\
kknudson@ufl.edu
\and Bei Wang \\ University of Utah\\beiwang@sci.utah.edu}
\date{}
\begin{document}

\maketitle 
\begin{abstract}

Inspired by the works of Forman on discrete Morse theory, which is a combinatorial adaptation to cell complexes of classical Morse theory on manifolds, we introduce a discrete analogue of the stratified Morse theory of Goresky and MacPherson.
We describe the basics of this theory and prove fundamental theorems relating the topology of a general simplicial complex with the critical simplices of a discrete stratified Morse function on the complex. 
We also provide an algorithm that constructs a discrete stratified Morse function out of an arbitrary function defined on a finite simplicial complex; 
this is different from simply  constructing a discrete Morse function on such a complex. 
We then give simple examples to convey the utility of our theory. 
Finally, we relate our theory with the classical stratified Morse theory in terms of triangulated Whitney stratified spaces. 

\end{abstract}

\pagestyle{plain}

%------------------------------------------------------------------------------------------------
\section{Introduction}
\label{sec:introduction}

It is difficult to overstate the utility of classical Morse theory in the study of manifolds. 
A Morse function $f:\Mspace \to \Rspace$ determines an enormous amount of information about the manifold $\Mspace$:~a handlebody decomposition, a realization of $\Mspace$ as a CW-complex whose cells are determined by the critical points of $f$, a chain complex for computing the integral homology of $\Mspace$, and much more. 

With this as motivation, Forman developed discrete Morse theory on general cell complexes~\cite{Forman1998}. This is a combinatorial theory in which function values are assigned not to points in a space but rather to entire cells. Such functions are not arbitrary; the defining conditions require that function values generically increase with the dimensions of the cells in the complex. Given a cell complex with set of cells $K$, a discrete Morse function $f: K \to \Rspace$ yields information about the cell complex similar to what happens in the smooth case.

While the category of manifolds is rather expansive, it is not sufficient to describe all situations of interest. Sometimes one is forced to deal with singularities, most notably in the study of algebraic varieties. One approach to this is to expand the class of functions one allows, and this led to the development of stratified Morse theory by Goresky and MacPherson~\cite{GoreskyMacPherson1988}. The main objects of study in this theory are {\em Whitney stratified spaces}, which decompose into pieces that are smooth manifolds. Such spaces are triangulable.

The goal of this paper is to generalize stratified Morse theory to finite simplicial complexes, much as Forman did in the classical smooth case. 
Given that stratified spaces admit simplicial structures, and any simplicial complex admits interesting discrete Morse functions, this could be the end of the story. 
However, we present examples in this paper illustrating that the class of discrete stratified Morse functions defined here is much larger than that of discrete Morse functions. Moreover, there exist discrete stratified Morse functions that are nontrivial and interesting from a data analysis point of view. 
Our motivations are three-fold. We address the first movitation in Section \ref{sec:pointdata}; the second and third are the subjects of ongoing research.

\begin{enumerate}%\denselist
\item \textbf{Generating discrete stratified Morse functions from point cloud data.}
Consider the following scenario. Suppose $K$ is a simplicial complex and that $f$ is a function defined on the $0$-skeleton of $K$. Such functions arise naturally in data analysis where one has a sample of function values on a space. Algorithms exist to build discrete Morse functions on $K$ extending $f$ (see, for example, \cite{KingKnudsonMramor2005}). Unfortunately, these are often of potentially high computational complexity and might not behave as well as we would like. In our framework, we may take this input and generate a discrete stratified  Morse function which will not be a global discrete Morse function in general, but which will allow us to obtain interesting information about the underlying complex.

\item \textbf{Filtration-preserving reductions of complexes in persistent homology and parallel computation.}
As discrete Morse theory is useful for providing a filtration-preserving reduction of complexes in the computation of both persistent homology~\cite{DlotkoWagner2012,MischaikowNanda2013,RobinsWoodSheppard2011} and multi-parameter
persistent homology~\cite{AlliliKaczynskiLandi2017}, we believe that  discrete stratified Morse theory could help to push the computational boundary even further. 
First, given any real-valued function defined on a simplicial complex, $f: K \to \Rspace$, our algorithm generates a stratification of $K$ such that the restriction of $f$ to each stratum is a discrete Morse function. Applying \emph{Morse pairing} to each stratum reduces $K$ to a  smaller complex of the same homotopy type. 
Second, if such a reduction can be performed in a filtration-preserving way with respect to each stratum, it would lead to a faster computation of persistent homology in the setting where the function is not required to be Morse. 
Finally, since discrete Morse theory can be applied independently to each stratum of $K$, we can design a parallel algorithm that computes persistent homology pairings by strata and uses the stratification, which captures relations among strata pieces, to combine the results. 

\item \textbf{Applications in imaging and visualization.}
Discrete Morse theory can be used to construct discrete Morse complexes in imaging (e.g.~\cite{Delgado-FriedrichsRobinsSheppard2015,RobinsWoodSheppard2011}), as well as Morse-Smale complexes~\cite{EdelsbrunnerHarerNatarajan2003,EdelsbrunnerHarerZomorodian2003} in visualization (e.g.~\cite{GuntherReininghausSeidel2014,GyulassyBremerPascucci2008}). 
In addition, it plays an essential role in the visualization of scalar fields and vector fields (e.g.~\cite{Reininghaus2012,ReininghausKastenWeinkauf2011}). 
Since discrete stratified Morse theory leads naturally to  stratification-induced domain partitioning where discrete Morse theory becomes applicable, we envision our theory to have wide applicability for the analysis and visualization of large complex data. 
\end{enumerate}

\para{Contributions.} 
Throughout the paper, we hope to convey via simple examples the usability of our theory. 
It is important to note that our discrete stratified Morse theory is \emph{not} a simple reinterpretation of  discrete Morse theory; it considers a larger class of functions defined on any finite simplicial complex and has potentially many implications for data analysis. 
Our contributions are: 
\begin{enumerate} %\denselist
\item We describe the basics of a discrete stratified Morse theory and prove fundamental theorems that relate the topology of a finite  simplicial complex with the critical simplices of a discrete stratified Morse function defined on the complex. 
\item We provide an algorithm that constructs a discrete stratified Morse function on any finite simplicial complex equipped with a real-valued function.
\item We prove that given a stratified set $S$ equipped with a triangulation $T$ and a stratified Morse function $f:S\to\zr$, there is an integer $r$ such that the $r$-th barycentric subdivision of $T$ supports a discrete stratified Morse function whose critical cells correspond to the critical points of $f$.
\item We demonstrate how to build a discrete stratified Morse function from a function defined on the vertices of a simplicial complex, based on a modification of the algorithm by King et al.~\cite{KingKnudsonMramor2005}; therefore addressing the first motivation. 
\end{enumerate}

An extended abstract of the present paper previously appeared as a conference paper~\cite{KnudsonWang2018}, which gave preliminary results surrounding contributions 1 and 2 above. 
The current paper contains the following extensions that encompass improvements of and changes to the conference version as well as new results.
In particular, we change the definition of a stratified simplicial complex (Definition~\ref{def:stratified-sc}) to be well-aligned with its continuous counterpart (e.g.~Whitney stratification) that considers the condition of the frontier. 
Given such a new definition, Theorems~\ref{thm:weak-dsmt-a} and~\ref{thm:weak-dsmt-b} discuss the change of homotopy type surrounding critical cells. 
We give new results that relate discrete Morse and discrete stratified Morse functions (Theorems~\ref{thm:dvf-union} and~\ref{thm:dmf-strata}).
We further characterize the coarseness property of our algorithm in constructing stratified Morse functions from any real-valued function on a simplicial complex (Proposition~\ref{prop:algcoarse}).
Finally, we discuss the applications of our theory to classical stratified Morse theory in discretizing a stratified Morse function (Theorem~\ref{thm:dsmftriangulation}) and provide an algorithm to generate discrete stratified Morse functions from point data (Theorem~\ref{thm:dsmfextdmf}). 

\para{A simple example.} 
We begin with an example from~\cite{Forman2002}, where we demonstrate how a discrete stratified Morse function can be constructed from a function that is not a discrete Morse function. 
As illustrated in Figure~\ref{fig:dsmf-circle-0}, the function on the left is a discrete Morse function where the green arrows can be viewed as its discrete gradient vector field; function $f$ in the middle is not a discrete Morse function, as the vertex $f^{-1}(5)$ and the edge $f^{-1}(0)$ both violate the defining conditions of a discrete Morse function. However, we can equip $f$ with a stratification $s$ by treating such violators as their own independent strata and taking care of boundary conditions, therefore converting $f$ into a discrete stratified Morse function.  

\begin{figure}[!ht]
 \begin{center}
  \includegraphics[width=0.7\linewidth]{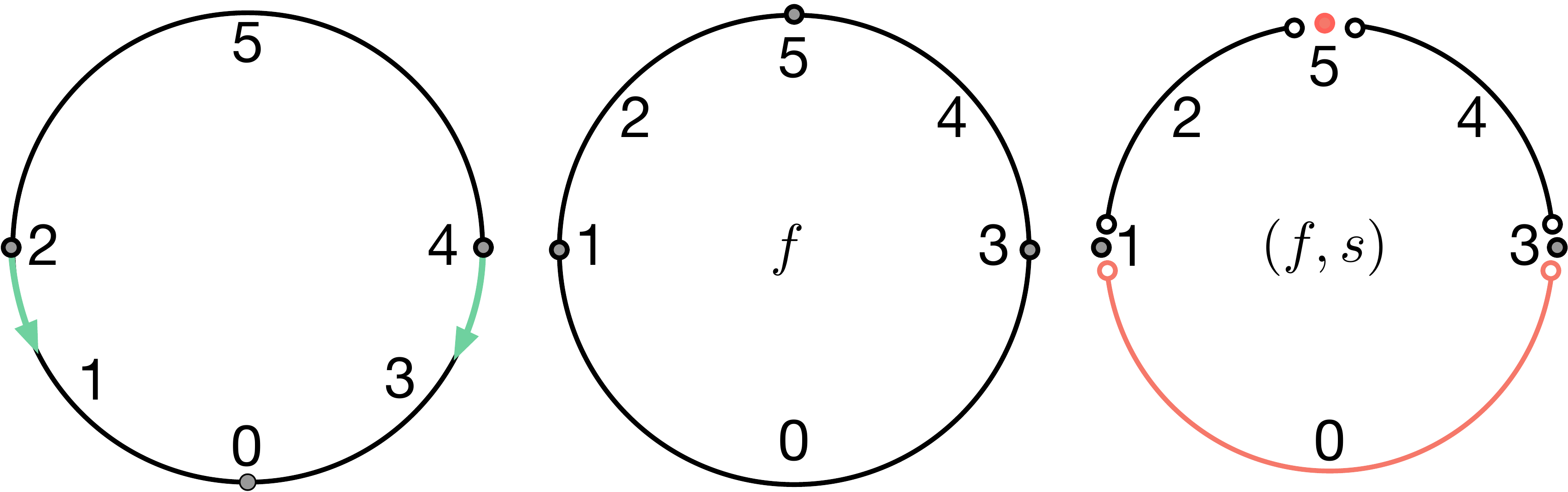}
 \end{center}
 %\vspace{-4mm}
\caption{The function on the left is a discrete Morse function. The function $f$  in the middle is not a discrete Morse function; however, it can be converted into a discrete stratified Morse function on the right when it is equipped with an appropriate stratification $s$.}
\label{fig:dsmf-circle-0}
\end{figure}

\section{Preliminaries on Discrete Morse Theory}
\label{sec:prelim}

We review the most relevant definitions and results on discrete Morse theory and refer the reader to Appendix \ref{sec:prelim-SMT} for a review of classical Morse theory. 
Discrete Morse theory is a combinatorial version of Morse theory \cite{Forman1998,Forman2002}. It can be defined for any CW complex but in this paper we will restrict our attention to finite simplicial complexes. 

\para{Discrete Morse functions.}
Let $K$ be any finite simplicial complex, where $K$ need not be a triangulated manifold nor have any other special property~\cite{Forman1999}. When we write $K$ we mean the set of simplices of $K$; by $|K|$ we mean the underlying topological space. 
Let $\alpha^{(p)} \in K$ denote a simplex of dimension $p$. 
Let $\alpha < \beta$ denote that simplex $\alpha$ is a face of simplex $\beta$. If $f:K\to \Rspace$ is a function
define $U(\alpha) = \{\beta^{(p+1)} > \alpha \mid f(\beta) \leq f(\alpha) \}$ 
and $L(\alpha) = \{\gamma^{(p-1)} < \alpha \mid f(\gamma) \geq f(\alpha) \}$.
In other words, $U(\alpha)$ contains the immediate cofaces of $\alpha$ with lower (or equal) function values, while $L(\alpha)$ contains the immediate  faces of $\alpha$ with higher (or equal) function values. 
Let $|U(\alpha)|$ and $|L(\alpha)|$ be their sizes.  

\begin{definition}\label{def:dmf}
A function $f: K \to \Rspace$ is a \emph{discrete Morse function} if for every $\alpha^{(p)} \in K$, 
(i) $|U(\alpha)| \leq 1$ and 
(ii) $|L(\alpha)| \leq 1$.
\end{definition}

Forman showed that conditions (i) and (ii) are exclusive -- if one of the sets $U(\alpha)$ or $L(\alpha)$ is nonempty then the other one must be empty (\cite{Forman1998}, Lemma 2.5).
Therefore each simplex $\alpha \in K$ can be paired with at most one exception simplex: either a face $\gamma$ with larger function value, or a coface $\beta$ with smaller function value. 
Formally, this means that 
if $K$ is a simplicial complex with a discrete Morse function $f$, then for any simplex $\alpha$, 
either (i) $|U(\alpha)| = 0$ or (ii) $|L(\alpha)| = 0$ (\cite{Forman2002}, Lemma 2.4).

\begin{definition}
A simplex $\alpha^{(p)}$ is \emph{critical} if (i) $|U(\alpha)| = 0$ and (ii) $|L(\alpha)| = 0$. A {\em critical value} of $f$ is its value at a critical simplex. 
\end{definition}

\begin{definition}
A simplex $\alpha^{(p)}$ is \emph{noncritical} if either of the following conditions holds: (i) $|U(\alpha)| = 1$; (ii) $|L(\alpha)| = 1$; as noted above these conditions can not both be true (\cite{Forman1998}, Lemma 2.5). 
\end{definition}

Given $c\in \Rspace$, we have the \emph{sublevel complex} $K_c = \cup_{f(\alpha) \leq c} \cup_{\beta \leq \alpha} \beta$. 
That is, $K_c$ contains all simplices $\alpha$ of $K$ such that $f(\alpha) \leq c$ along with all of their faces.  

\para{Results.}
We have the following two combinatorial versions of the main results of classical Morse theory. 
\begin{theorem}[DMT Theorem A,~\cite{Forman1999}] 
\label{theorem:dmt-a}
Suppose the interval $(a,b]$ contains no critical value of $f$. 
Then $K_b$ is homotopy equivalent to $K_a$. 
In fact, $K_b$ simplicially collapses onto $K_a$.
\end{theorem}

A key component in the proof of Theorem~\ref{theorem:dmt-a} is the following fact~\cite{Forman1998}: for a simplicial complex equipped with an arbitrary discrete Morse function, when passing from one sublevel complex to the next, the noncritical simplices are added in pairs, each of which consists of a simplex and a free face. 

The next theorem explains how the topology of the sublevel complexes changes as one passes a critical value of a discrete Morse function. In what follows, $\dot{e}^{(p)}$ denotes the boundary of a $p$-simplex $e^{(p)}$. 

\begin{theorem}[DMT Theorem B,~\cite{Forman1999}] 
\label{theorem:dmt-b}
Suppose $\sigma^{(p)}$ is a critical simplex with $f(\sigma) \in (a,b]$, and there are no other critical simplices with values in $(a,b]$. Then $K_b$ is homotopy equivalent to the space obtained by attaching a $p$-cell $e^{(p)}$ along its entire boundary in $K_a$; that is, $K_b= K_a \cup_{\dot{e}^{(p)}} e^{(p)}$.
\end{theorem}

\para{The associated gradient vector field.} Given a discrete Morse function $f:K\to\Rspace$ we may associate a discrete gradient vector field as follows. Since any noncritical simplex $\alpha^{(p)}$ has at most one of the sets $U(\alpha)$ and $L(\alpha)$ nonempty, there is a unique face $\nu^{(p-1)}<\alpha$ with $f(\nu)\ge f(\alpha)$ or a unique coface $\beta^{(p+1)}>\alpha$ with $f(\beta)\le f(\alpha)$. Denote by $V$ the collection of all such pairs $\{\sigma <\tau\}$. Then every simplex in $K$ is in at most one pair in $V$ and the simplices not in any pair are precisely the critical cells of the function $f$. We call $V$ the {\em gradient vector field associated to $f$}. We visualize $V$ by drawing an arrow from $\alpha$ to $\beta$ for every pair $\{\alpha<\beta\}\in V$. Theorems \ref{theorem:dmt-a} and \ref{theorem:dmt-b} may then be visualized in terms of $V$ by collapsing the pairs in $V$ using the arrows. Thus a discrete gradient (or equivalently a discrete Morse function) provides a collapsing order for the complex $K$, simplifying it to a complex $L$ with potentially fewer cells but having the same homotopy type.

The collection $V$ has the following property. By a \emph{$V$-path}, we mean a sequence $$\alpha^{(p)}_0 < \beta_0^{(p+1)}>\alpha_1^{(p)}<\beta_1^{(p+1)}>\cdots <\beta_r^{(p+1)}>\alpha_{r+1}^{(p)}$$ where each $\{\alpha_i<\beta_i\}$ is a pair in $V$. Such a path is {\em nontrivial} if $r>0$ and {\em closed} if $\alpha_{r+1}=\alpha_0$. Forman proved the following result.

\begin{theorem}[\cite{Forman1998}] If $V$ is a gradient vector field associated to a discrete Morse function $f$ on $K$, then $V$ has no nontrivial closed $V$-paths.
\end{theorem}

In fact, if one defines a discrete vector field $V$ to be a collection of pairs of simplices of $K$ such that each simplex is in at most one pair in $V$, then one can show that if $V$ has no nontrivial closed $V$-paths there is a discrete Morse function $f$ on $K$ whose associated gradient is precisely $V$. 

We note here the following result that will be needed below. A proof may be found in \cite[p.~99]{Knudson2015}.

\begin{lemma}\label{lem:refinedmf} Suppose $K'$ is the barycentric subdivision of $K$ and let $V$ be a discrete gradient vector field on $K$. Then there is a discrete gradient vector field $V'$ on $K'$ such that the critical cells of $V$ and $V'$ are in one-to-one correspondence. In fact, for a critical $p$-cell $\alpha \in K$ of $V$, one may choose a $p$-cell $\alpha' \in K'$ which is critical for $V'$. \hfill 
\end{lemma}

\section{A Discrete Stratified Morse Theory}
\label{sec:theory}

Our goal is to describe a combinatorial version of  stratified Morse theory. 
To do so, we need to: (a) define a discrete stratified Morse function; and (b) prove the combinatorial versions of the relevant fundamental results. 
Our results are very general as they apply to any finite simplicial complex $K$ equipped with a real-valued function 
$f: K \to \Rspace$. 
Our work is motivated by relevant concepts from (classical) stratified Morse theory~\cite{GoreskyMacPherson1988}, whose details are found in Appendix \ref{sec:prelim-SMT}.

\subsection{Background}
\label{subsec:dsmf-background}

\para{Open simplices.}
To state our main results, we need to consider open simplices (as opposed to the closed simplices of Section~\ref{sec:prelim}).  
Let $\{a_0, a_1, \cdots, a_k\}$ be a geometrically independent set in $\Rspace^N$, a \emph{closed $k$-simplex} $[\sigma]$ is the set of all points $x$ of $\Rspace^N$ such that $x = \sum_{i=0}^{k}t_ia_i$, where $\sum_{i=0}^{k}t_i = 1$ and $t_i \geq 0$ for all $i$~\cite{Munkres1984}. 
An \emph{open simplex} $(\sigma)$ is the interior of the closed simplex $[\sigma]$.

A \emph{simplicial complex} $K$ is a finite set of open simplices such that: 
(a) If $(\sigma) \in K$ then all open faces of $[\sigma]$ are in $K$; 
(b) If $(\sigma_1), (\sigma_2) \in K$ and $(\sigma_1) \cap (\sigma_2) \neq \emptyset$, then $(\sigma_1) = (\sigma_2)$. 
For the remainder of this paper, we always work with a finite open simplicial complex $K$.  Unless otherwise specified, we work with open simplices $\sigma$ and define the boundary $\dot{\sigma}$ to be the boundary of its closure. 

\para{Stratified simplicial complexes.} 
In the conference version of this paper~\cite{KnudsonWang2018}, we worked with a weak notion of stratification. We have since discovered technical issues with that definition; work in progress seeks to find the most general setting in which our theory can be applied. In this paper,  we employ Definition~\ref{def:stratified-sc}, which resembles the definition of $\Pset$-decomposition for $\Pset$ being a partially ordered set (poset)  in~\cite[p.~36]{GoreskyMacPherson1988}. 
Recall a subset $S$ of a topological space $Z$ is \emph{locally closed} if it is the intersection of an open and a closed set in $Z$. For a topological space $S$, let $\overline{S}$ denote its closure, $\mathring{S}$ its interior. 

\begin{definition}
\label{def:poset-stratification}
Let $\Pset$ be a poset with order relation denoted by $<$.
A \emph{$\Pset$-decomposition} of a topological space $Z$ is a locally finite collection of disjoint locally closed subsets called \emph{pieces}, $S_i \subset Z$ (one for each $i \in \Pset$) such that $Z = \bigcup_{i \in \Pset} S_i$,  and $S_i \cap \overline{S_j} = \emptyset$ if and only if $S_i \subset \overline{S_j}$. 
\end{definition}

We now define a stratified simplicial complex as follows. 

\begin{definition}
\label{def:stratified-sc} 
A {\em stratification} of $K$, $\Scal = \{S_i\}$, is a locally finite collection of disjoint locally closed subsets called \emph{strata}, $S_i \subset K$, such that $K = \bigcup S_i$, and which satisfies the condition of the frontier: $S_i\cap \overline{S}_j\ne \emptyset$ if and only if $S_i \subset \overline{S_j}$; equivalently, the frontier $\overline{S_i} \setminus \mathring{S_i}$ of each $S_i$ is a union of strata. Each $S_i$ is a union of (open) simplices; its connected components are called \emph{strata pieces}. 
\end{definition}

The condition of the frontier in Definition~\ref{def:stratified-sc} imposes a partial order $\Pset = (\Scal, <)$ on the strata: $S_i <  S_j$ if and only if $S_i \subset \overline{S_j}$. 
\begin{lemma}
\label{lem:minimal-element} 
A minimal element in the partial order $\Pset = (\Scal, <)$ is a subcomplex of $K$.
\end{lemma}

\begin{proof} Suppose $S_i$ is such a minimal element and suppose $\sigma \in S_i$. It suffices to show that $\partial\sigma\in S_i$ as well. Suppose $\tau<\sigma$. Then $\tau\in S_j$ for some $j$ and so $\tau\in S_j\cap \overline{S_i}$. This implies that $S_j\subseteq \overline{S_i}$. But $S_i$ is minimal in the order $\Pset$ and so $S_j=S_i$; that is, $\tau \in S_i$.
%%\qed
\end{proof}

A stratification gives an \emph{assignment} from $K$ to the set $\Scal$, denoted $s: K \to \Scal$.
In our setting, each $S_i$ is the union of finitely many open simplices (that may not form a subcomplex of $K$); 
and each open simplex $\sigma$ in $K$ is assigned to a particular stratum $s(\sigma)$ via the mapping $s$. Since these subspaces may not be simplicial complexes, we must modify Definition~\ref{def:dmf} as follows.

\begin{definition}
\label{def:dmf-on-strata}
Suppose $S_i$ is a stratum in $\Scal$. A function $f:S_i\to \Rspace$ is a {\em discrete Morse function} if for every $\alpha^{(p)}\in S_i$, (i) $|U(\alpha)|\le 1$, (ii) $|L(\alpha)|\le 1$, and (iii) if one of the sets $U(\alpha)$ or $L(\alpha)$ is nonempty then the other must be empty.
\end{definition}

Condition (iii) above is not necessary for functions defined on simplicial complexes, but the proof of that relies on the fact that all faces of a simplex are in the complex as well. This need not be true for the various strata and so we impose the condition here.

\para{Stratum-preserving homotopies.}
If $\Xspace$ and $\Yspace$ are two stratified spaces, we call a map $f: \Xspace \to \Yspace$ \emph{stratum-preserving} if the
image of each component of a stratum of $\Xspace$ lies in a stratum of $\Yspace$~\cite{Friedman2003}.
Such a map $f: \Xspace \to \Yspace$ is a \emph{stratum-preserving homotopy equivalence} if there exists a stratum-preserving map $g: \Yspace \to \Xspace$ such that $g\circ f$ and $f \circ g$ are stratum-preserving homotopic to the identity~\cite{Friedman2003}.

\subsection{Discrete Stratified Morse Function}
\label{subsec:dsmf-primer}

\para{Discrete stratified Morse function.}
Let $K$ be a simplicial complex equipped with a stratification $s$ and a function $f: K \to \Rspace$. 
We define 
\begin{align}
U_s(\alpha) & = \{\beta^{(p+1)} > \alpha \mid s(\beta)=s(\alpha)\;\text{and}\;f(\beta) \leq f(\alpha) \}, \nonumber\\
L_s(\alpha) & = \{\gamma^{(p-1)} < \alpha \mid s(\gamma)=s(\alpha)\;\text{and}\; f(\gamma) \geq f(\alpha) \}. \nonumber
\end{align}

\begin{definition}
\label{def:dsmf}
Given a simplicial complex $K$ equipped with a stratification $s: K \to \Scal$,
a function $f: K \to \Rspace$ (equipped with $s$) is a \emph{discrete stratified Morse function} if for every $\alpha^{(p)} \in K$, 
(i) $|U_s(\alpha)| \leq 1$, 
(ii) $|L_s(\alpha)| \leq 1$, and
(iii) if one of these sets is nonempty then the other must be empty.
\end{definition}
In other words, a discrete stratified Morse function is a pair $(f,s)$ where $f:K\to\Rspace$ is a discrete Morse function when restricted to each stratum $S_j\in\Scal$ (in the sense of Definition \ref{def:dmf-on-strata}). We omit the symbol $s$ whenever it is clear from the context.

\begin{definition}
A simplex $\alpha^{(p)}$ is {\em globally critical} if $|U(\alpha)|=|L(\alpha)|=0$.
A simplex $\alpha^{(p)}$ is \emph{locally critical} if it is not globally critical and if  $|U_s(\alpha)| = |L_s(\alpha)|= 0$. A {\em critical value} of $f$ is its value at a critical simplex. 
\end{definition}

\begin{definition}
\label{def:dsmf-noncritical}
A simplex $\alpha^{(p)}$ is {\em globally noncritical} if $|U(\alpha)| + |L(\alpha)| = 1$. A simplex $\alpha^{(p)}$ is \emph{locally noncritical} if it is not globally noncritical and
exactly one of the following two conditions holds: (i) $|U_s(\alpha)| = 1$ and $|L_s(\alpha)| = 0$; or (ii) $|L_s(\alpha)| = 1$ and $|U_s(\alpha)| = 0$.
\end{definition}

The two conditions in Definition~\ref{def:dsmf-noncritical} mean that, within the same stratum as $s(\alpha)$: (i) there is a $\beta^{(p+1)} > \alpha$ with $f(\beta) \leq f(\alpha)$ or (ii) there is a $\gamma^{(p-1)} < \alpha$ with $f(\gamma) \geq f(\alpha)$; conditions (i) and (ii) cannot both be true.

A classical discrete Morse function $f:K\to\Rspace$ is a discrete stratified Morse function with the trivial stratification ${\mathcal S} = \{K\}$. We will present several examples in Section~\ref{sec:examples} illustrating that the class of discrete stratified Morse functions is much larger.

\para{Violators.}
The following definition is central to our algorithm in constructing a discrete stratified Morse function from any real-valued function defined on a simplicial complex. 

\begin{definition}
\label{definition:violator}
Suppose $K$ is a simplicial complex equipped with a real-valued function $f: K \to \Rspace$.
A simplex $\alpha^{(p)}$ is a \emph{violator} of the conditions associated with a discrete Morse function if one of these conditions holds: (i) $|U(\alpha)| \geq 2$; (ii) $|L(\alpha)| \geq 2$; (iii) $|U(\alpha)| = 1$ and $|L(\alpha)| = 1$. 
These are referred to as type I, II and III violators; the sets containing such violators are not necessarily mutually exclusive. 
\end{definition}

Here is a useful fact about violators that we shall need later.

\begin{lemma}\label{lem:violator-boundary}
Suppose $(f,s):K\to\Rspace$ is a discrete stratified Morse function. If $\alpha$ is a violator for $f$, then either $\alpha$ is locally critical or $\alpha$ is a boundary simplex for the stratification; that is, either some face $\nu$ of $\alpha$ is in the frontier of $s(\alpha)$ or $\alpha$ is in the frontier of the stratum $s(\tau)$ of a coface $\tau$.
\end{lemma}

\begin{proof}
By definition, $f|_{s(\alpha)}$ is a discrete Morse function on $s(\alpha)$. It is possible that $\alpha$ is critical for this restriction and since $\alpha$ is a violator it cannot be globally critical. Otherwise, there is either a face $\nu<\alpha$ with $f(\nu)\ge f(\alpha)$ or a coface $\tau>\alpha$ with $f(\alpha)\ge f(\tau)$, and this paired simplex ($\nu$ or $\tau$) also lies in $s(\alpha)$. But $\alpha$ is a global violator. So in either case, there is another face $\nu'<\alpha$ or coface $\tau'>\alpha$ causing the violation. But $\nu',\tau'\not\in s(\alpha)$ and hence either $\nu'$ belongs to the frontier of $s(\alpha)$ or $\alpha$ belongs to the frontier of $s(\tau')$ (and hence $s(\alpha)\subset \overline{s(\tau')}$).
%\qed
\end{proof}

\subsection{Back and Forth: Discrete Morse and Discrete Stratified Morse Functions}
\label{subsec:backandforth}
An honest discrete Morse function $f:K\to\Rspace$ is a discrete stratified Morse function for the trivial stratification $\Scal=\{K\}$. More is true however.

\begin{lemma}
\label{lem:dmf-restrict} 
Suppose $f:K\to\Rspace$ is a discrete Morse function and let $\Scal =\{S_i\}$ be a stratification of $K$, with $s:K\to\Scal$ the assignment map. Then $(f,s)$ is a discrete stratified Morse function.
\end{lemma}

\begin{proof} Since $f$ is a discrete Morse function, for every simplex $\alpha$ the sets $U(\alpha)$ and $L(\alpha)$ satisfy the required conditions. In particular, if one of them is nonempty then the other is empty. Since $U_s(\alpha)\subseteq U(\alpha)$ and $L_s(\alpha)\subseteq L(\alpha)$, the conditions of Definition \ref{def:dsmf} hold.
%\qed
\end{proof}

Lemma~\ref{lem:dmf-restrict} is in contrast with the smooth case. Indeed, a Morse function on a manifold $M$ may not be a stratified Morse function on an arbitrary stratification of $M$. For example, for a torus equipped with the standard height function $h$, choose a regular value $c$ such that $h^{-1}(c)$ consists of two disjoint circles $C_1$ and $C_2$. Take the stratification of the torus consisting of a point on $C_1$, the circle $C_1$, and the complement of $C_1$. Then $h$ is {\em not} a stratified Morse function with respect to this stratification since $h|_{C_1}$ is constant. However, a small perturbation of $h$ is a stratified Morse function.

Lemma \ref{lem:dmf-restrict} is {\em not} true for discrete gradient vector fields, however. Suppose $V$ is a discrete gradient on $K$ associated to some function $f:K\to\Rspace$ and suppose $\Scal = \{S_i\}$ is a stratification. It is entirely possible that a regular simplex $\alpha$ is paired with a simplex $\beta$ with $s(\alpha)\ne s(\beta)$. That is, the vector field $V$ may be orthogonal to the strata. We do have the following result.

\begin{theorem}
\label{thm:dvf-union} 
Suppose $(f,s):K\to\Rspace$ is a discrete stratified Morse function with stratification $\Scal=\{S_i\}$. For each $i$, denote by $V_i$ the discrete gradient vector field associated to $f|_{S_i}$, and let $V = \bigcup_i V_i$. Then $V$ is a discrete gradient vector field on $K$.
\end{theorem}

\begin{proof} It suffices to show that there are no closed $V$-paths. Suppose $$\gamma := \{\alpha_0<\beta_0>\alpha_1<\beta_1>\cdots >\alpha_t<\beta_t>\alpha_0\}$$ is a closed $V$-path. Then $\gamma$ is not contained in a single stratum piece. Say it lies in two strata pieces: $\{\alpha_0<\beta_0>\alpha_1<\beta_1>\cdots <\beta_k\} \subset S_1$, $\{\alpha_{k+1}<\beta_{k+1}>\cdots <\beta_u\} \subset S_2$, and $\{\alpha_{u+1}<\beta_{u+1}>\cdots <\beta_t>\alpha_0\} \subset S_1$. Note that $\gamma$ must decompose this way since simplices can be paired only within the same stratum piece. Since $\beta_k>\alpha_{k+1}$, we have $\alpha_{k+1}\in \overline{S_1}$ and so by the frontier condition we have $S_2\subset \overline{S_1}$. Also, since $\beta_u>\alpha_{u+1}$, we have $\alpha_{u+1}\in \overline{S_2}$ and again the frontier condition implies $S_1\subset \overline{S_2}$. It follows that $\overline{S_1}=\overline{S_2}$ and since strata pieces are disjoint we conclude that $S_1=S_2$; that is, $\gamma$ lies in a single stratum piece, a contradiction. The general case of $\gamma$ passing through multiple strata pieces follows inductively.
%\qed
\end{proof}

Of course, the vector field $V$ is {\em not} associated to the function $f$; that is, it is not the gradient vector field of $f$ (more on this later). 
The gradient field $V$ produced in Theorem~\ref{thm:dvf-union} respects the strata in the sense that each pair $\{\alpha<\beta\}$ in $V$ satisfies $s(\alpha)=s(\beta)$. The following result is a useful technical tool for us in the sequel.

\begin{theorem}
\label{thm:dmf-strata}
Suppose $(f,s):K\to\Rspace$ is a discrete stratified Morse function with stratification $\Scal=\{S_i\}$, and let $V$ be the induced discrete gradient vector field on $K$. If necessary, extend the partial order on $\Scal$ to a linear order and write the strata $S_1<\cdots <S_n$. Then there is a discrete Morse function $g:K\to\Rspace$ satisfying the following properties.
\begin{enumerate} %\denselist
    \item The gradient of $g$ is $V$.
    \item There are real numbers $a_1< a_2 < \cdots <a_n$ such that $g^{-1}(-\infty,a_i] = \displaystyle \bigcup_{j\le i} S_j$ for $1 \leq i \leq n$.
\end{enumerate}
\end{theorem}

\begin{proof}
There are infinitely many discrete Morse functions compatible with $V$; we need only construct one satisfying the second property. The standard way to construct discrete Morse functions with gradient $V$ is to consider the Hasse diagram of $K$, modified by reversing arrows from $\beta$ to $\alpha$ whenever $\{\alpha<\beta\}$ is a pair in $V$. This is an acyclic directed graph and a standard result in graph theory is that such graphs support functions on their vertices whose function values decrease along every directed path. Such a function yields a discrete Morse function on $K$ with gradient $V$. We know that the minimal element $S_1$ is a subcomplex of $K$ (Lemma \ref{lem:minimal-element}); choose a discrete Morse function $g_1$ on $S_1$ compatible with $V_1$ and set $a_0=\max_{\sigma\in S_1}g_1(\sigma)$. Assume inductively that we have constructed $g_i$, a discrete Morse function on $S_1\cup\cdots \cup S_i$ satisfying the second property. We extend it to $S_{i+1}$ as follows. Collapse the subgraph of the Hasse diagram corresponding to $S_1\cup\cdots \cup S_i$ to a point. This is then a sink in this directed graph. Since $f|_{S_{i+1}}$ is a discrete Morse function, we can find a function $g_{i+1}$ whose gradient agrees with $V_{i+1}$ on $S_{i+1}$ and which satisfies $g_{i+1}(\sigma)> a_i$ for all $\sigma\in S_{i+1}$. Set $a_{i+1}= \max_{\sigma\in S_{i+1}} g_{i+1}(\sigma)$. This completes the inductive step.
%\qed
\end{proof}

We say that a function $g$ satisfying the conclusions of Theorem \ref{thm:dmf-strata} {\em separates} the strata. 

\subsection{Homotopy Type}
\label{subsec:homotopytype}

In both smooth and discrete Morse theory, we have theorems about how the topology of the sublevel sets (or sublevel complexes, in the discrete case) vary as we move through increasing function values. The same is true in stratified Morse theory, where a neighborhood of a critical point consists of {\em Morse data}, which is a product of tangential and normal data (see Appendix \ref{sec:prelim-SMT}). Our definition of a discrete stratified Morse function is too loose to allow for such theorems as it stands. The issue is that we have no control on how the function values change as we cross from one stratum to another, as opposed to the smooth case where the function is continuous and so function values cannot vary too much in a neighborhood of a critical point.

We can still say something, however, in the form of Theorems~\ref{thm:weak-dsmt-a} and~\ref{thm:weak-dsmt-b}. 

\begin{theorem}[Weak DSMT Theorem A]
\label{thm:weak-dsmt-a}
Given a discrete stratified Morse function $(f,s)$, performing a collapse of either a global noncritical pair or a local noncritical pair is a stratum-preserving homotopy equivalence.
\end{theorem}

\begin{proof}
We make use of the auxiliary discrete Morse function constructed in Theorem \ref{thm:dmf-strata}. Suppose $(f,s):K\to\Rspace$ is a discrete stratified Morse function with associated discrete gradient $V$. Let $g$ be a discrete Morse function with gradient $V$ which separates the strata. Then any noncritical pair, global or local, is simply a regular pair for $g$. By Theorem \ref{theorem:dmt-a} we may collapse such a pair without changing the homotopy type of the complex. Moreover, since all noncritical pairs lie within a stratum, this homotopy equivalence is stratum-preserving.
%\qed
\end{proof}

Describing what happens around a critical cell is much more complicated. We discuss this further below, but for now we can say the following. A consequence of Theorems \ref{theorem:dmt-a} and \ref{theorem:dmt-b} is that if the simplicial complex $K$ has a discrete gradient vector field $V$, then $K$ has the homotopy type of a CW-complex with one cell for each critical cell of $V$ (of the same dimension). Theorem \ref{thm:dvf-union} then implies the following result.

\begin{corollary}[Weak DSMT Theorem B]
\label{thm:weak-dsmt-b} 
Suppose $(f,s):K\to\Rspace$ is a discrete stratified Morse function and denote by $V$ the discrete gradient vector field obtained as the union of the $V_i$ associated with $f|_{S_i}$. Then $K$ has the homotopy type of a CW-complex with one cell for each critical cell of $V$.
\end{corollary}

\subsection{Algorithm for Constructing Discrete Stratified Morse Functions}
\label{subsec:dsmf-algorithm}
We give an algorithm to construct a discrete stratified Morse function from any real-valued function on a simplicial complex, $f: K \to \Rspace$ as follows. 
\begin{enumerate}%\denselist
\item Make a single pass of all simplices in $K$, order the violators $\Vcal = \{\sigma_1,\sigma_2,\dots,\sigma_r\}$ by increasing dimension and by increasing function value within each dimension. 
\item Initialize $\Scal = \emptyset$, $i=1$. 
\item Remove $\sigma_i$ from $\Vcal$ and add $\sigma_i$ to $\Scal$.  
\item Consider $K_i=K\setminus \{\sigma_1,\dots,\sigma_i\}$: 
\begin{itemize} %\denselist
\item If the restriction of $f$ to $K_i$, $f|_{K_i}$,  is a discrete Morse function, then let $J$ denote the set of indices $k\le i$ such that $\sigma_k\in \overline{K_i}$ and add the following strata to $\Scal$ (which may contain more than two strata pieces): the frontier $\overline{K_i}\setminus (\mathring{K_i} \cup \{\sigma_j\}_{j\in J})$ and $\mathring{K_i}$. 
\item Otherwise, if $f|_{K_i}$ is not a discrete Morse function, then at least one $\sigma_j$ with $j>i$ remains a violator.
\end{itemize}
\item Remove simplices that are no longer violators from the list and repeat the steps 2-4 above until no violators are left.  
\end{enumerate}

\begin{lemma}
\label{lem:algo-strat} 
The collection $\Scal$ satisfies the condition of the frontier and therefore meets the conditions of Definition~\ref{def:stratified-sc}.
\end{lemma}

\begin{proof} First note that every simplex in $K$ belongs to some strata piece; the strata pieces are obviously disjoint. If $\sigma_k$ and $\sigma_\ell$ are distinct violators in $\Scal$, then $\sigma_k\cap \overline{\sigma_\ell}$ if and only if $\sigma_k\in\partial\sigma_\ell$. Similarly, if $\sigma_k$ intersects the closure of the frontier strata piece, then it must lie in the boundary of one of the simplices in that strata piece. A violator in $\Scal$ cannot intersect the open strata piece  $\mathring{K_i}$ by definition, and if it intersects its closure then it intersects the frontier strata piece and we are done.
%\qed
\end{proof}

\begin{theorem}
\label{theorem:algorithm}
The function $(f,s)$ associated to the stratification defined in the algorithm above is a discrete stratified Morse function.
\end{theorem}

\begin{proof}
We assume $K$ is connected. 
If $f$ itself is a discrete Morse function, then there are no violators in $K$. The algorithm produces the trivial stratification $\Scal = \{K\}$ and since $f$ is a discrete Morse function on the entire complex, the pair $(f,s)$ trivially satisfies Definition~\ref{def:dsmf}.

If $f$ is not a discrete Morse function, let $\Scal = {\Vcal} \cup \{F\}\cup \{I\}$ denote the stratification produced by the algorithm, where ${\mathcal V}$ is the set of violators that form their own strata, $F$ is the set of frontier strata pieces and $I$ is the interior complementary to $F$. Since each violator $\alpha \in \Vcal$ forms its own strata piece  $s(\alpha)$, the  restriction of $f$ to $s(\alpha)$ is trivially a discrete Morse function in which $\alpha$ is a critical simplex. 

Recall that the sets $F$ and $I$ are obtained as follows. We remove the collection ${\mathcal V}$ from $K$ to obtain $L$ and consider the restriction of $f$ to this subspace. The function $f$ is a discrete Morse function here, and since $I$ is the interior of $L$, $f$ restricts to a discrete Morse function on $I$. The set $F$ is obtained as $\overline{L} \setminus (I\cup {\mathcal V})$. If $\sigma$ is a simplex in $F$, then $\sigma$ is not one of the violators for $f$ that get removed (or it is not a violator at all in the first place).  It follows that $f|_F$ is a discrete Morse function and we are done.
%\qed
\end{proof}

\begin{remark} When we restrict the function $f:K\to\Rspace$ to one of the strata $S_i \in \Scal$, a non-violator $\sigma$ that is regular globally (that is, $\sigma$  forms a gradient pair with a unique simplex $\tau$) may become a critical simplex for the restriction of $f$ to $S_i$. 
\end{remark}

The algorithm is relatively efficient. Suppose $K$ has $n$ simplices and let $c$ be the maximum number of codimension-1 faces and cofaces of any simplex in $K$ (in other words, $c$ could be considered as the maximum ``degree" of a simplex in $K$). The first step of the algorithm takes $O(cn)$ steps to identify the set of violators $\Vcal$ by checking for each simplex $\alpha^{(p)}$, its faces $\gamma^{(p-1)} < \alpha$ and cofaces $\beta^{(p+1)} > \alpha$ . Then for each violator $\sigma_i$ removed from the set $\Vcal$, the algorithm must check the complex $K_i$ for remaining violators by paying attention to simplicies adjacent to $\sigma_i$, which takes $O(c)$. Since the number of violators $r$ is at most $n$, this requires $O(cn)$ as well. So the algorithm runs in $O(cn)$ time. If $c$ is a constant, then the algorithm runs linear in the number of simplices. 

\subsection{Coarseness}\label{subsec:coarse}

Suppose $f$ is a function on $K$ and denote the set of stratifications $\Scal$ of the complex $K$ on which $f$ is a discrete stratified Morse function by $\Sigma(K, f)$.  The set $\Sigma(K, f)$ is partially ordered by inclusion: $\Scal\le \Scal'$ if each stratum piece $S_i \in \Scal$ is contained in some element of $\Scal'$. Generally, we wish to work with coarse stratifications; that is, we seek maximal elements of $\Sigma(K,f)$. Our algorithm in Section~\ref{subsec:dsmf-algorithm}  does just that.

\begin{proposition}\label{prop:algcoarse}
The stratification produced by the algorithm of Section~\ref{subsec:dsmf-algorithm} is a maximal element of $\Sigma(K,f)$.
\end{proposition}

\begin{proof}
The algorithm produces a stratification $\Scal$ consisting of some violators $\sigma_1,\sigma_2,\dots ,\sigma_\ell$ for the function $f$, the interior of $K_\ell = K\setminus\{\sigma_1,\dots ,\sigma_\ell\}$, and the frontier of $K_\ell$ (with the violators removed). Suppose there is a stratification $\Scal'\in\Sigma(K,f)$ with $\Scal\le\Scal'$. Then there is some stratum piece $S \in \Scal'$ containing $\mathring{K_\ell}$. If they are not equal, then there is a simplex $\alpha$ in $S \setminus \mathring{K_\ell}$. If $\alpha$ is one of the violators for $f$ then $f|_S$ cannot be a discrete Morse function on $S$, otherwise the algorithm would have terminated sooner. If $\alpha$ lies in the frontier of $K_\ell$ then $S$ contains the entire frontier by definition and hence must be all of $K_\ell$. But then $\Scal'$ would not satisfy the frontier condition. Thus, $\mathring{K_\ell}$ must be one of the elements of $\Scal'$. Similarly, the frontier of $K_\ell$ must also be an element of $\Scal'$. The violators are disjoint and so they must then also be elements of $\Scal'$. It follows that $\Scal=\Scal'$ and that $\Scal$ is maximal.
%\qed
\end{proof}

Note that if $f$ is a discrete Morse function on $K$, then the algorithm produces the trivial stratification $\Scal = \{K\}$, which is indeed maximal in $\Sigma(K,f)$.

\section{Discrete Stratified Morse Theory by Example}
\label{sec:examples}

We apply the algorithm described in Section~\ref{subsec:dsmf-algorithm} to a collection of examples to demonstrate the utility of our theory. 
For each example, given a function $f: K \to \Rspace$ that is not necessarily a discrete Morse function, 
we equip $f$ with a particular stratification $s$, thereby converting it to a discrete stratified Morse function $(f,s)$.
These examples help to illustrate that the class of discrete stratified Morse functions is much larger than that of discrete Morse functions.  

\begin{figure}[h!]
 \begin{center}
  \includegraphics[width=0.7\linewidth]{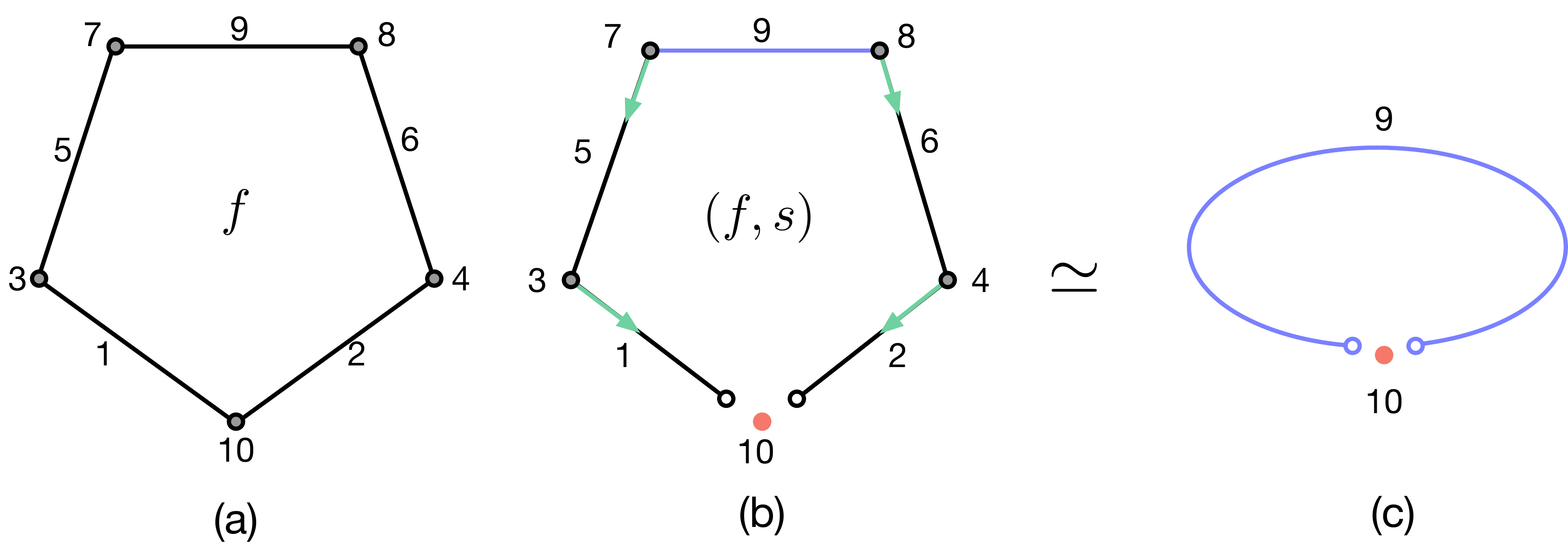}
 \end{center}
 %\vspace{-6mm}
\caption{Upside-down pentagon. (a): $f$ is not a discrete Morse function. (b): $(f,s)$ is a discrete stratified Morse function where violators  removed by the algorithm are in red; critical simplicies (other than the violators) are in blue; the discrete gradient vector field is marked by green arrows. (c): the simplified simplicial complex by removing the Morse pairs following the discrete gradient vector field.}
\label{fig:dsmf-upside-down-pentagon}
\end{figure}

\para{Example 1:~upside-down pentagon.}
As illustrated in Figure~\ref{fig:dsmf-upside-down-pentagon} (a), $f: K \to \Rspace$ defined on the boundary of an upside-down pentagon is not a discrete Morse function, as it contains a set of violators: 
$\Vcal = \{f^{-1}(10), f^{-1}(1), f^{-1}(2)\}$, since  $|U(f^{-1}(10))| = 2$ and $|L(f^{-1}(1))| = |L(f^{-1}(2))| = 2$, respectively. 

By following the algorithm in Section~\ref{subsec:dsmf-algorithm}, we would first remove the violator $f^{-1}(10)$ and check to see if what remains is a discrete Morse function. 
We see that this is indeed the case: we have four Morse pairs illustrated by green arrows in Figure~\ref{fig:dsmf-upside-down-pentagon} (b): $(f^{-1}(3),f^{-1}(1))$, $(f^{-1}(4),f^{-1}(2))$, $(f^{-1}(7),f^{-1}(5))$, and $(f^{-1}(8),f^{-1}(6))$.  
The resulting discrete stratified Morse function $(f,s)$ is a discrete Morse function when restricted to each stratum. 
Recall that a simplex is critical for $(f,s)$ if it is neither the source nor the target of a discrete gradient vector. The critical values of $(f,s)$ are therefore $9$ and $10$. 

One of the primary uses of classical discrete Morse theory is {\em simplification}. In this example, we can collapse a portion of each stratum following the discrete gradient field (illustrated by green arrows, see Section~\ref{sec:prelim}). Removing the Morse pairs simplifies the original complex as much as possible without changing its homotopy type, and the resulting simplification yields a complex with one vertex and one edge, see Figure~\ref{fig:dsmf-upside-down-pentagon} (c).

\para{Example 2:~pentagon.}
For our second pentagon example, $f$ can be made into a discrete stratified Morse function $(f,s)$ by making $f^{-1}(0)$ (a type II violator) and $f^{-1}(9)$ (a type I violator) their own strata following the algorithm in Section~\ref{subsec:dsmf-algorithm} (Figure~\ref{fig:dsmf-pentagon}).  
The critical values of $(f,s)$ are $0, 1, 3, 7, 8$ and $9$.
It is important to note that $f^{-1}(1)$ and $f^{-1}(3)$ are considered critical as they form their own strata pieces; however they are not the violators removed by the algorithm. 
The simplicial complex can be reduced to one with fewer cells by canceling the Morse pairs, as shown in Figure~\ref{fig:dsmf-pentagon} (d).   
\begin{figure}[h!]
 \begin{center}
  \includegraphics[width=0.9\linewidth]{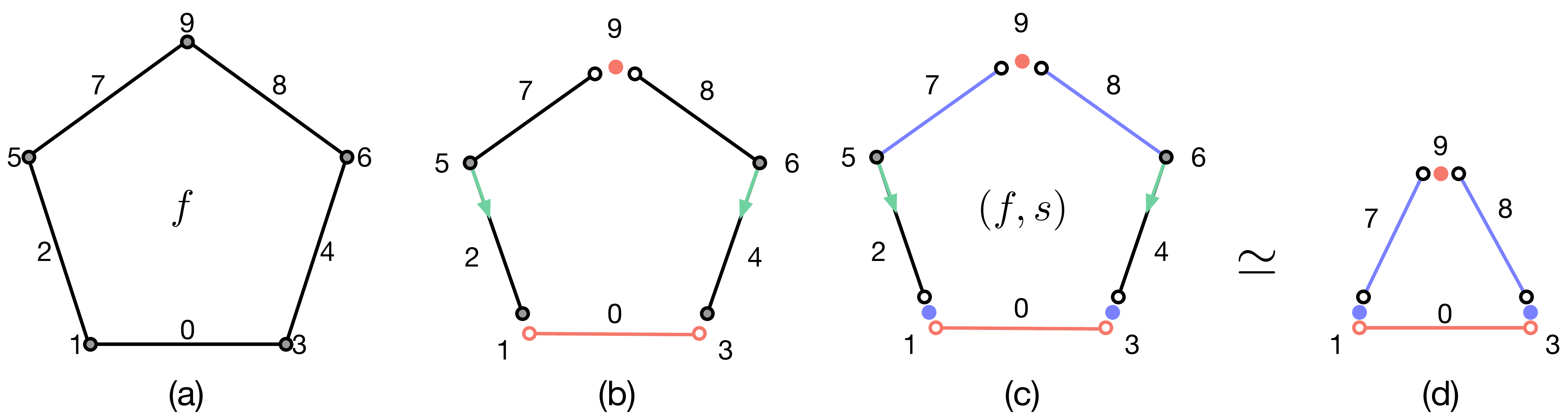}
 \end{center}
 %\vspace{-6mm}
\caption{Pentagon. (a): $f$ is not a discrete Morse function. (b): an intermediate simplicial complex after removing violators in red. (c): separating simplicies $f^{-1}(1)$ and $f^{-1}(3)$ from (b) results in a stratification that satisfies the frontier condition. There are six strata pieces associated with the discrete stratified Morse function $(f,s)$. (d): the simplified simplicial complex. }
\label{fig:dsmf-pentagon}
\end{figure}

\para{Example 3: split octagon.}
The split octagon example (Figure~\ref{fig:dsmf-split-octagon}) begins with a function $f$ defined on a triangulation of a stratified space that consists of two $0$-dimensional and three $1$-dimensional strata pieces  (Figure~\ref{fig:dsmf-split-octagon}(a)). The set of violators to be considered is $\Vcal = \{f^{-1}(0)$, $f^{-1}(10)$, $f^{-1}(24)$, $f^{-1}(30)$, $f^{-1}(31)\}$.  
However, after removing $f^{-1}(30)$, then $f^{-1}(31)$, the rest of the simplicies in $\Vcal$ are no longer violators and the restriction of $f$ to what is left is a discrete Morse function (Figure~\ref{fig:dsmf-split-octagon}(b)). 
The result of canceling Morse pairs yields a simpler complex shown in Figure~\ref{fig:dsmf-split-octagon}(c).
\begin{figure}[h!]
 \begin{center}
  \includegraphics[width=0.8\linewidth]{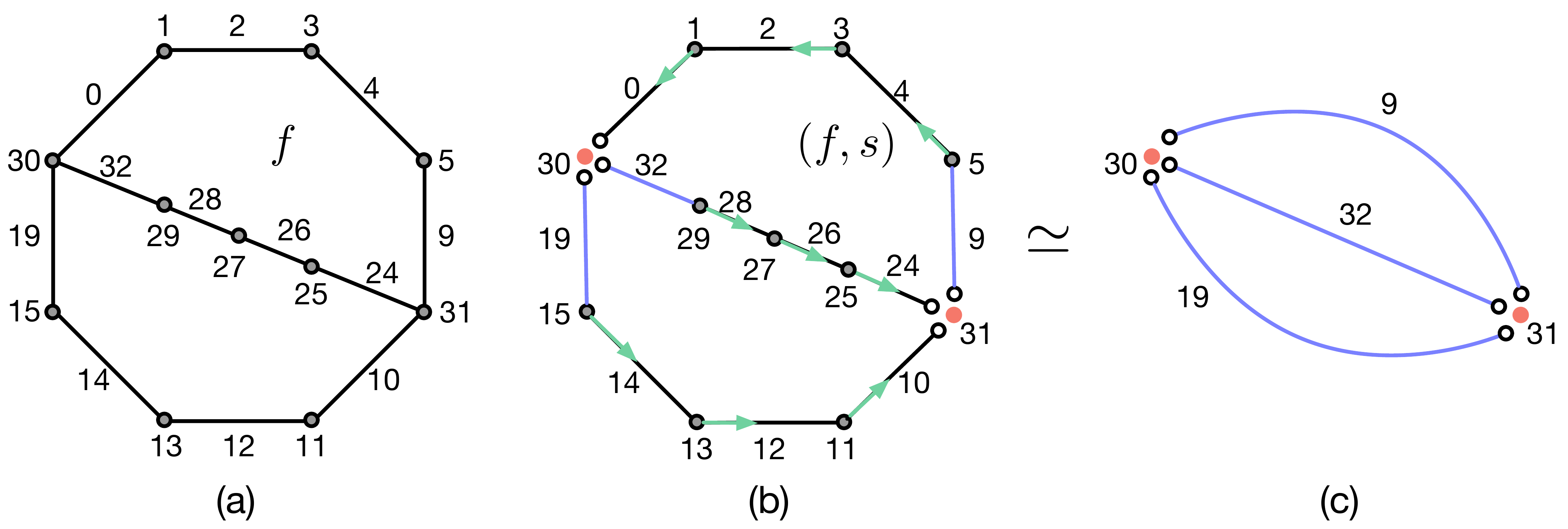}
 \end{center}
 %\vspace{-6mm}
\caption{Split octagon. (a) $f$ is defined on the triangulation of a stratified space. (b) the resulting discrete stratified Morse function $(f,s)$. (c) the simplified complex. }
\label{fig:dsmf-split-octagon}
\end{figure}

\para{Example 4: tetrahedron.}
In Figure~\ref{fig:dsmf-tetra}(a), the values of the function $f$ defined on the simplices of a tetrahedron are specified for each dimension. 
For each simplex $\alpha \in K$, we list the elements of its corresponding $U(\alpha)$ and $L(\alpha)$ in Table~\ref{table:tetra}. We also classify each simplex in terms of its criticality in the setting of classical discrete Morse theory. 
\begin{table}[!ht]
\begin{center}
\resizebox{\textwidth}{!}{
\begin{tabular}{|c |c|c|c|c|c|c|c|c|c|c|c|c|c|c| } 
 \hline
 & 1 & 2 & 3 & 4 & 5 & 6 & 7 & 8 & 9 & 10 & 11 & 12 & 13 & 14 \\ 
 $U(\alpha)$ & $\emptyset$ &$\emptyset$ &$\{2\}$ &$\emptyset$ &
 $\emptyset$ & $\emptyset$ & $\{5\}$& $\{6\}$ &$\emptyset$ & $\{4,7\}$ & $\{6\}$  &$\{9\}$  & $\emptyset$& $\{8, 11,12\}$\\  
 $L(\alpha)$ & $\emptyset$ & $\{3\}$&$\emptyset$ & $\{10\}$ 
 & $\{7\}$& $\{8, 11\}$ & $\{10\}$ &$\{14\}$ & $\{12\}$& $\emptyset$
 & $\{14\}$  & $\{14\}$  & $\emptyset$&$\emptyset$ \\
 Type & C & R & R & R & R & II & III & III & R & I & III & III &C & I\\
 \hline
\end{tabular}
}
\end{center}
%\vspace{-4mm}
\caption{Tetrahedron. For simplicity, a simplex $\alpha$ is represented by its function value $f(\alpha)$ (as $f$ is 1-to-1). In terms of criticality for each simplex: C means critical; R means regular; I, II and III correspond to type I, II and III violators. }
\label{table:tetra}
\end{table}
According to Table~\ref{table:tetra} the violators $\Vcal$ have function values of $10, 14$ (type I), $6$ (type II), $7, 8, 11$ and $12$ (type III). 

\begin{figure}[h!]
 \begin{center}
  \includegraphics[width=0.8\linewidth]{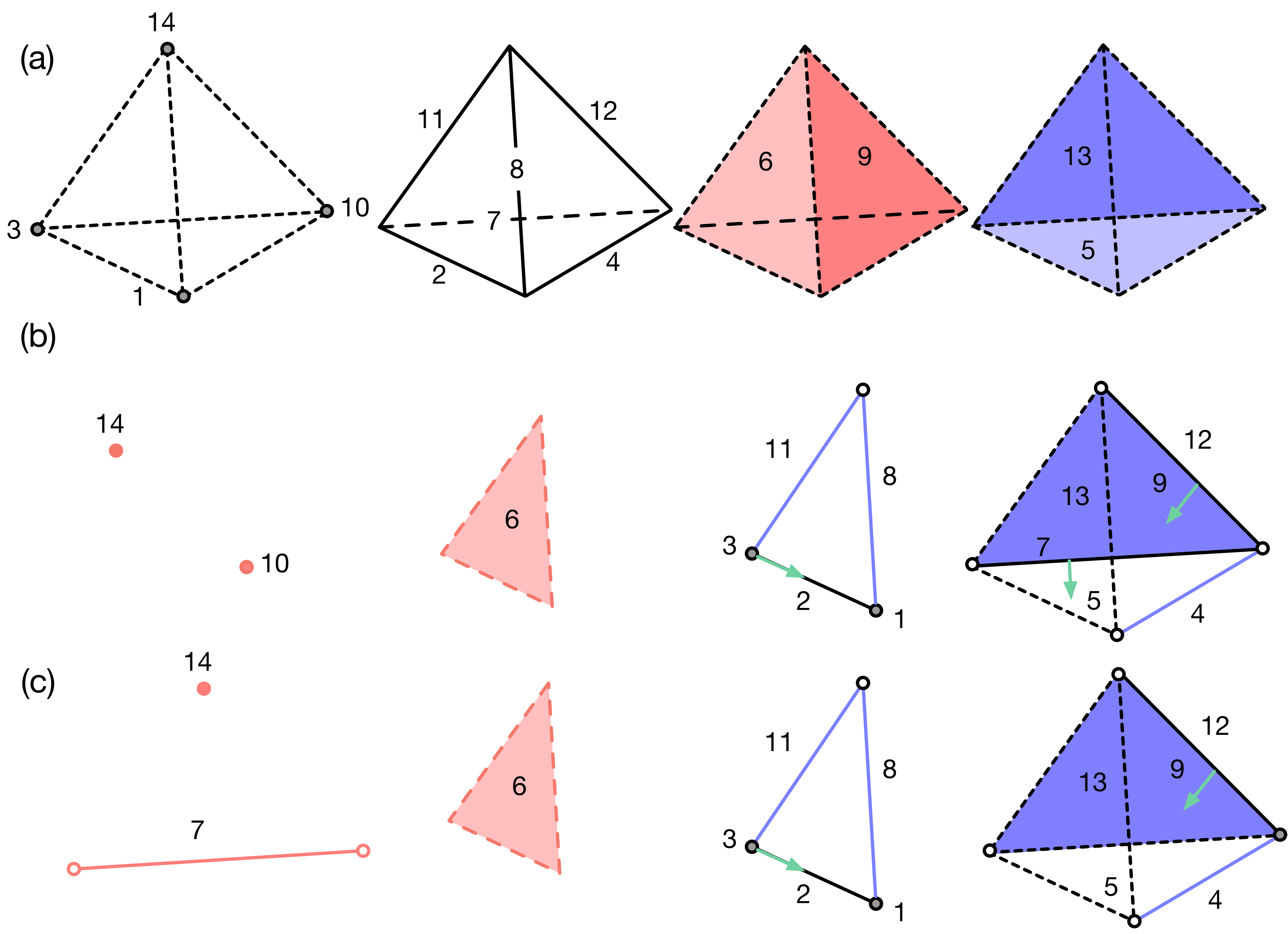}
 \end{center}
 %\vspace{-6mm}
\caption{Tetrahedron. (a) $f$ is defined on the simplices of increasing dimensions. (b) the resulting discrete stratified Morse function  $(f,s)$ is shown by individual strata pieces; using the algorithm in Section~\ref{subsec:dsmf-algorithm}. Not all simplices are shown. (c) An alternative stratification.}
\label{fig:dsmf-tetra}
\end{figure}

We describe our algorithm step by step, the intermediate results (strata pieces) are illustrated in Figure~\ref{fig:dsmf-tetra}(b). 
For simplicity, a simplex $\alpha$ is represented by its function value $f(\alpha)$. 
First, initialize $\Scal = \emptyset$.  
Second, remove the vertex $10$, then $7$ is no longer a violator, remove it from the list $\Vcal$;  now $\Scal = \{ 10 \}$. 
Third, remove the vertex $14$, then $8, 11, 12$ are no longer violators, remove them from the list $\Vcal$; $\Scal = \{ 10, 14\}$. 
Fourth, $6$ is the only remaining violator, add it to $\Scal = \{10, 14, 6\}$. 
Finally, let $C = K \setminus \{10, 14, 6\}$. Then $\overline{C} = K \setminus \{6\}$ and $\mathring{C} = C \setminus \{2, 8, 11, 13\}$. Add $\mathring{C}$ and $\overline{C} \setminus (\mathring{C} \cup \{14\}) $ to $\Scal$. 
$\Scal$ now contains 5 strata pieces. 
Besides vertices $10$ and $14$ and triangle $6$, $\Scal$ also contains a strata piece $\{1,2,3,8,11\}$ that is homotopy equivalent to an open 1-manifold; vertex 3 and edge 2 forms a Morse pair. 
The last strata piece in $\Scal$ is $\{4,5,7,9,12,13\}$, which is topologically a punctured disc; in particular, there are two Morse pairs, $(12, 9)$ and $(7,5)$. 

As an alternative to the algorithm described in Section~\ref{subsec:dsmf-algorithm}, we show in Figure~\ref{fig:dsmf-tetra}(c) that we could obtain a different stratification by changing the ordering of the violators to be removed. 
As in (b), $\Vcal = \{10, 14, 6, 7, 8, 11, 12\}$. 
First, initialize $\Scal = \emptyset$.  
Second, remove the vertex $14$, then $8, 11, 12$ are no longer violators, remove them from $\Vcal$;  now $\Scal = \{ 14 \}$. 
Third, remove the edge $7$, then $10$ is no longer a violator, remove it from $\Vcal$; 
$\Scal = \{ 14, 7\}$. 
Fourth, $6$ is the only remaining violator, add it to $\Scal = \{14, 7, 6\}$. 
Finally, let $C = K \setminus \{14, 7, 6\}$. 
Then $\overline{C} = K \setminus \{6\}$ and $\mathring{C} = C \setminus \{1, 2, 3, 8, 11\}$.  
$\Scal$ now contains 5 strata pieces in (c) that are slightly different from (b).  Note that the stratifications in (b) and (c) are incomparable in the set $\Sigma(K,f)$.

\para{Example 5: split solid square.}
As illustrated in Figure~\ref{fig:dsmf-split-square}, the function $f$ defined on a split solid square is not a discrete Morse function; there are three type I violators $f^{-1}(9)$, $f^{-1}(10)$, and $f^{-1}(11)$. 
Making these violators their own strata (in the order of increasing function value following the algorithm in Section~\ref{subsec:dsmf-algorithm}) helps to convert $f$ into a discrete stratified Morse function $(f,s)$. In this example, all simplices are considered critical for $(f,s)$. For instance, consider the open 2-simplex $f^{-1}(4)$, we have $L(f^{-1}(4)) = \{f^{-1}(11)\}$ and $U(f^{-1}(4)) = \emptyset$; with the stratification $s$ in Figure~\ref{fig:dsmf-split-square} (right),  $L_s(f^{-1}(4)) = \emptyset$ and so $4$ is not a critical value for $f$ but it is a critical value for $(f,s)$. Since every simplex is critical for $(f,s)$, there is no simplification to be done.
\begin{figure}[h!]
 \begin{center}
  \includegraphics[width=0.5\linewidth]{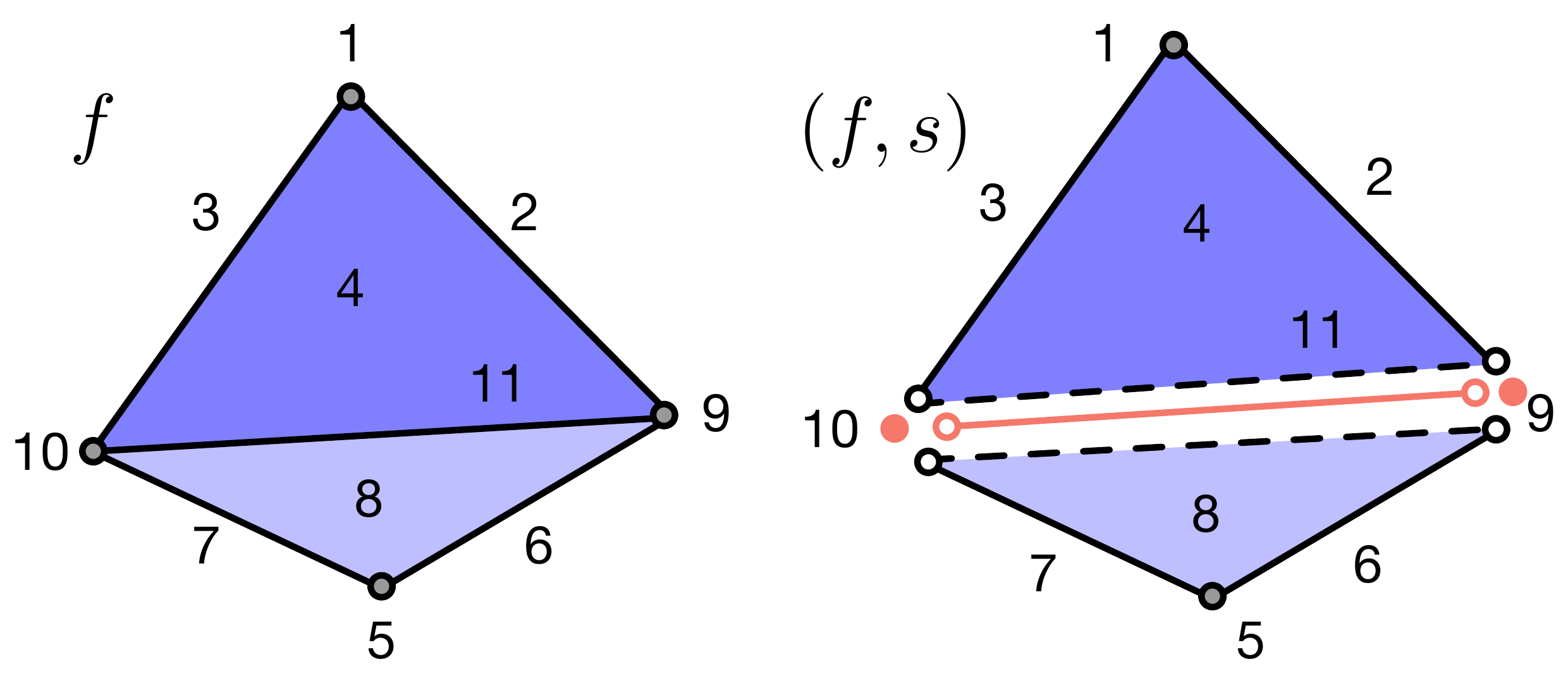}
 \end{center}
 %\vspace{-6mm}
\caption{Split solid square. Every simplex is critical for $(f,s)$.}
\label{fig:dsmf-split-square}
\end{figure}

\section{Applications to Triangulations of Stratified Spaces}
\label{sec:triangulations}

\subsection{Background on Whitney Stratifications and Triangulations}
\label{subsec:whitney} 

\para{Whitney stratifications.}
We review relevant background on Whitney stratifications; the primary reference for the material in this section is~\cite{Mather2012}. For simplicity, we assume all manifolds are smooth (i.e., of class $C^\infty$). If $x,y\in \Rspace^n$ with $x\ne y$, we define the {\em secant} $\overline{xy}$ to be the line through the origin in $\Rspace^n$ parallel to the line joining $x$ and $y$. If $x\in\Rspace^n$, we identify the tangent space $T_x\Rspace^n$ with $\Rspace^n$ in the standard way.

Let $M$ be a smooth manifold without boundary and let $Z$ be a subset of $M$. A \emph{stratification} $\Scal=\{S_i\}_{i \in \Pset}$ of $Z$ is a cover of $Z$ by pairwise disjoint smooth submanifolds of $M$ which lie in $Z$; these submanifolds $S_i$ are called \emph{strata} (whose connected components are referred to as \emph{strata pieces}); where $\Pset$ is some poset. 
The stratification $\Scal$ is \emph{locally finite} if each point of $M$ has a neighborhood which meets finitely many strata. We say $\Scal$ satisfies the condition of the \emph{frontier} if the strata in $\Scal$ satisfy $S_i\cap\overline{S_j}\ne\emptyset$ if and only if $S_i\subseteq \overline{S_j}$; or equivalently, if for each stratum $S_i$ of $\Scal$ its frontier $(\overline{S_i}\setminus S_i)\cap Z$ is a union of strata. Compare with Definitions~\ref{def:poset-stratification} and~\ref{def:stratified-sc}.   

\begin{definition}\label{conditionab} Let $X$ and $Y$ be submanifolds of a smooth manifold $M$. We say that $X$ is {\em Whitney regular} over $Y$ if whenever $\{x_i\} \subset X$ and $\{y_i\} \subset Y$ are sequences of points both converging to some point $y\in Y$, the lines $\ell_i = \overline{x_iy_i}$ converge to a line $\ell \in \Rspace^n$, and the tangent spaces $T_{x_i}X$ converge to a space $T \subseteq \Rspace^n$, then
\begin{enumerate}%\denselist
\item[(A)] $T_yY\subseteq T$, and
\item[(B)] $\ell\subseteq T$.
\end{enumerate}
\end{definition}

\begin{remark} Convergence here should be thought of as taking place in a small neighborhood of $y$ identified with $\Rspace^n$ via a coordinate chart. Also, Condition B above implies Condition A \cite{Mather2012}.
\end{remark}

\begin{proposition}~\cite[Proposition 2.7]{Mather2012}
\label{frontier} 
Suppose $y\in\overline{X\setminus Y}$ and $(X,Y)$ satisfies condition B at $y$. Then $\dim Y<\dim X$.
\end{proposition}

\begin{definition}
\label{whitneystrat} 
A stratification $\Scal=\{S_i\}$ is a {\em Whitney stratification} if it is locally finite, satisfies the condition of the frontier, and if whenever $j\le i$, $S_i$ is Whitney regular over $S_j$.
\end{definition}

\begin{remark} Let $\Scal$ be a Whitney stratification of a subset $Z$ of a manifold $M$ and let $S_i,S_j$ be strata. Proposition \ref{frontier} implies that if $i \le j$ then $\dim S_i<\dim S_j$. 
\end{remark}

Here is a useful way of constructing stratified spaces \cite{Johnson1983}. A stratified set of type 0 is a smooth manifold. To construct a stratified set $X^{(k+1)}$ of type $k+1$, take a stratified set $X^{(k)}$ of type $k$, a smooth manifold $K$, a smooth submanifold $L$ of codimension $0$ in $\partial K$, and an ``attaching'' map $\alpha:L\to X^{(k)}$, and then set $X^{(k+1)} = X^{(k)}\cup_\alpha K$. These attaching maps are not arbitrary continuous maps; they must be proper, continuous, and as close as possible to a smooth fiber bundle. The strata are then the various $X_k=X^{(k)}\setminus X^{(k-1)}$. For the pinched torus in Figure \ref{fig:pinched}, we begin with $X^{(0)}$ as the pinch point. To build $X^{(1)}$ we take the closed interval $K=[0,1]$, $L=\{0,1\}$, and $\alpha:L\to X^{(0)}$ the obvious map. To build $X^{(2)}$, we take $K$ to be the disjoint union of a disc and a square, $L$ to be the disjoint union of the boundary circle and boundary square, and $\alpha:L\to X^{(1)}$ to be the map identifying the circle via the identity and the square via the map that first yields a wedge of two circles and then collapses one to the base point.

\para{Triangulating stratified sets.}
By a {\em triangulation} of a set $Z$ we mean a finite simplicial complex $K$ and a homeomorphism $h:|K|\to Z$, where $|K|$ denotes the geometric realization of $K$. Any smooth manifold is triangulable, for example.

Suppose we have a compact set $Z$ with Whitney stratification $\Scal = \{S_i\}$. As above, we may think of $Z$ as being built up by the pieces $S_i$ in such a way that when $i \le j$, we have $S_i\subseteq \partial\overline{S_j}$ (this is essentially the condition of the frontier). We now have the following theorem (see Theorem 2.1 of \cite{Johnson1983} or Proposition 5 of \cite{Goresky1978}). 

\begin{theorem}~\cite[Theorem 2.1]{Johnson1983}
\label{thm:strat-triangulate} 
A compact Whitney stratified set $Z$ admits a triangulation by a finite simplicial complex so that each $\overline{S_i}$ is triangulated as a subcomplex.
\end{theorem}

\begin{figure}[!ht]
\centerline{\includegraphics[width=0.8\textwidth]{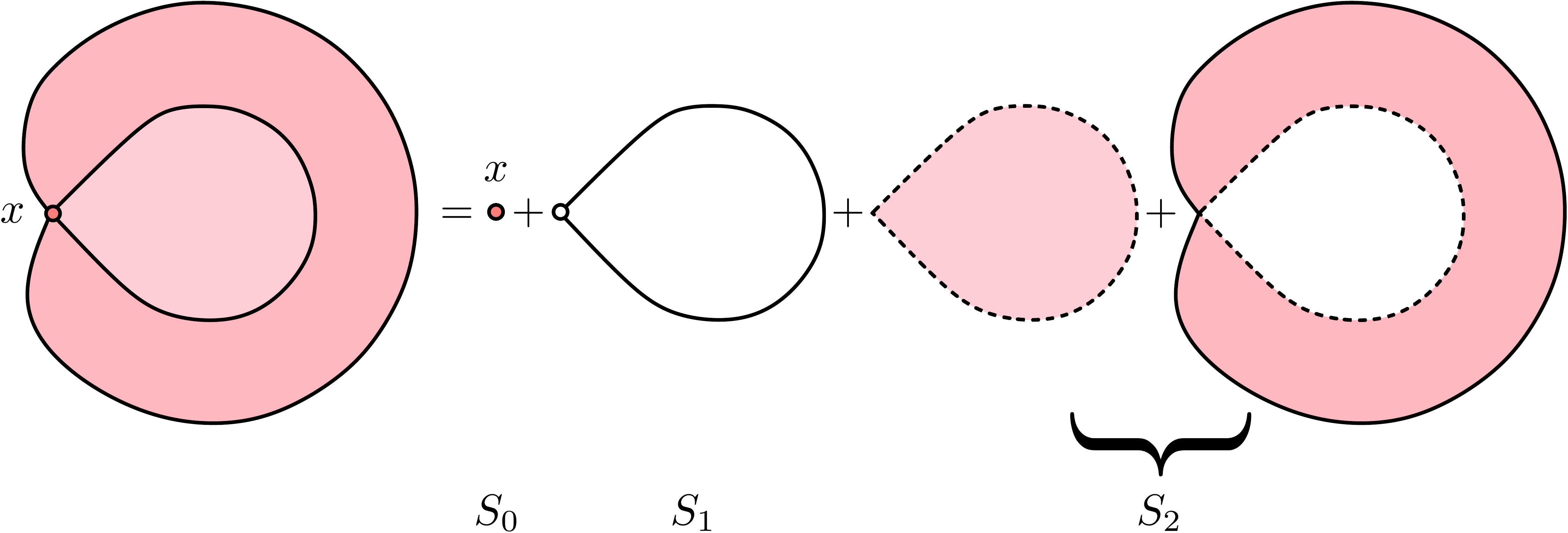}}
\vspace{-2mm}
\caption{The pinched torus as a stratified space, whose stratification is formed by strata $S_0$, $S_1$, and $S_2$.}
\label{fig:pinched}
\end{figure}

In addition, whenever $i\le j$, we have that $S_i$ is triangulated as a subcomplex of $\partial\overline{S_j}$. So, for example, in the pinched torus of Figure~\ref{fig:pinched} we have that the pinch point $x:=S_0$ is a vertex in the triangulation of $\overline{S_1}$, and that $\overline{S_1}$ is a subcomplex of the closure $\overline{S_2}$ of the two disjoint discs. The important idea here is that one may think of building the triangulation from the bottom up by first triangulating the $0$-stratum, then extending that to a triangulation of the $1$-stratum, and so on, noting that at each stage, the lower-dimensional (closed) stratum is a subcomplex of the boundary of the next (closed) stratum.

\subsection{Applications to Classical Stratified Morse Theory}
\label{subsec:smooth-discrete}

Suppose $Z$ is a Whitney stratified subset of a smooth manifold $M$ with stratification $\Scal = \{S_i\}$. A {\em stratified Morse function} $f:S\to\Rspace$ is, roughly speaking, a function that restricts to a Morse function on each stratum (see Appendix \ref{sec:prelim-SMT} for the formal definition). In this section, we investigate the following obvious question. Suppose $f:Z\to \Rspace$ is a stratified Morse function. Is there a triangulation of $Z$ and a discrete stratified Morse function on that triangulation that ``mirrors'' the behavior of $f$? That is, can we find a discrete stratified Morse function and a bijection between its critical cells and the critical points of the function $f$?

\para{Comparing classical (smooth) and discrete Morse theory.}\label{subsec:benedetti} To answer this question, we first need to address it in the classical nonstratified case. This has been solved satisfactorily by Benedetti \cite{Benedetti2012,Benedetti2016}.
Suppose $M$ is a smooth $d$-manifold with boundary (possibly empty) and $f:M\to\Rspace$ is a Morse function. Denote by $c_i$ the number of critical points of $f$ of index $i$. We call the $d$-tuple ${\mathbf c} = (c_0,c_1,\dots,c_d)$ the {\em Morse vector} of the function $f$ and we say that {\em $M$ admits ${\mathbf c}$ is a Morse vector}. A classical theorem of Morse asserts that the manifold $M$ is homotopy equivalent to a cell complex with $c_i$ cells of dimension $i$.

Similarly, if $K$ is a $d$-dimensional simplicial complex with a discrete Morse function $g:K\to\Rspace$ having $c_i$ critical cells of dimension $i$, we call ${\mathbf c} = (c_0,c_1,\dots,c_d)$ the {\em discrete Morse vector} of the function $g$ and say that {\em $K$ admits ${\mathbf c}$ as a discrete Morse vector}. If $K$ is a triangulation of a manifold $M$ with boundary, we say the function $g$ is {\em boundary critical} if all the cells in the subcomplex triangulating $\partial M$ are critical for $g$. Forman proved the analogue of Morse's theorem: the complex $K$ has the homotopy type of a cell complex with $c_i$ cells of dimension $i$. 
We recall Theorem 2.28 of~\cite{Benedetti2016} below. 

\begin{theorem}~\cite[Theorem 2.28]{Benedetti2016}
\label{thm:discsmoothmorse} 
If a smooth $d$-manifold $M$ {\em(}with boundary{\em)} admits ${\mathbf c}$ as a Morse vector, then for any PL triangulation $T$ of $M$, there exists an integer $r$ so that the $r$-th barycentric subdivision of $T$ admits 
\begin{enumerate}%\denselist
\item[{\em(a)}] a discrete Morse function with $c_i$ critical $i$-faces, and
\item[{\em(b)}] a boundary-critical discrete Morse function with $c_{d-i}$ critical interior $i$-faces. %\hfill $%\qed$
\end{enumerate}
\end{theorem}

The statement (b) in Theorem~\ref{thm:discsmoothmorse} is related to duality. If $f:M\to\Rspace$ is a Morse function on $M$ with Morse vector ${\mathbf c} = (c_0,c_1,\dots,c_d)$, then the function $-f:M\to\Rspace$ is also a Morse function but with Morse vector ${\mathbf c}^* = (c_d,c_{d-1},\dots,c_0)$. In the discrete case, the negative of a discrete Morse function on a complex $K$ is {\em not} a discrete Morse function. However, in the case of a triangulation $T$ of a manifold, one may consider the dual block complex $T^*$ with a corresponding dual function $f^*$, yielding an analogous result.

\para{Discretizing a stratified Morse function.} 
Suppose $Z$ is a compact set with stratification $\Scal=\{S_i\}$ and that $f:Z\to\Rspace$ is a stratified Morse function. Let $d$ denote the dimension of the top stratum. Set $d_i = \dim S_i$ and denote by ${\mathbf c}^i = (c_0^i,\dots,c_{d_i}^i)$ the Morse vector of $f|_{S_i}$. According to Theorem \ref{thm:strat-triangulate}, there is a triangulation $T$ of $Z$ so that each closed stratum $\overline{S_i}$ is triangulated as a subcomplex $T_i$. This leads to our main result in this section relating discrete stratified Morse theory to (classical) stratified Morse theory. 

\begin{theorem}
\label{thm:dsmftriangulation} 
There exists an integer $r$ such that the $r$-th barycentric subdivision of $T$ admits a discrete stratified Morse function $F$ satisfying the following:
\begin{enumerate}%\denselist
\item[{\em (a)}] the stratification of $T$ is given by the various $T_i\setminus T_{i-1}$, $i=0,\dots,d$; and
\item[{\em (b)}] the restriction of $F$ to the $i$-th stratum has discrete Morse vector ${\mathbf c}_i^* = (c_{d_i}^i,\dots,c_0^i)$.
\end{enumerate}
\end{theorem}

\begin{proof} Keeping in mind the discussion at the end of Section \ref{subsec:whitney}, we proceed as follows. The $0$-stratum $S_0$ is a smooth manifold. By Theorem \ref{thm:discsmoothmorse} we may choose $r_0$ so that the $r_0$-th subdivision of $T_0$ admits a (boundary critical) discrete Morse function with discrete Morse vector ${\mathbf c}_0^*$. (In this case we could also find a discrete Morse function with vector ${\mathbf c}_0$ since $S_0$ has no boundary, but this is not true moving forward). We now proceed inductively. Suppose the result is true for stratum $i \ge 0$, and consider the $r$-th subdivision of $T$, where $r=r_0+\cdots +r_i$. This means that we have stratified $T_i$ by the various $T_i\setminus T_{i-1}$ and we have a discrete stratified Morse function $F_i$ satisfing condition (b) above on $T_i$. We know that $S_{i}\subseteq \overline{S}_{i+1}$; in fact, it lies inside the boundary of $\overline{S}_{i+1}$. Again by Theorem \ref{thm:discsmoothmorse}, there is an integer $r_{i+1}$ so that the $r_{i+1}$-th subdivision of $T_{i+1}$ has a boundary critical discrete Morse function with vector ${\mathbf c}_{i+1}^*$. Observe that this requires subdivision of the subcomplex $T_i$, but by Lemma \ref{lem:refinedmf}, this subdivision of $T_i$ supports a discrete Morse function with the same Morse vector. This completes the proof.
%\qed
\end{proof}

\section{Generating Discrete Stratified Morse Functions from Point Data}\label{sec:pointdata}

In this section, we solve our first motivational problem in Section~\ref{sec:introduction}. 
That is, \textbf{given a simplicial complex $K$ equipped with an injective function on its vertices $f: K_0 \to \Rspace$, can we extend $f$ to a discrete stratified Morse function $\tilde{f}$ on $K$}? 

An algorithm to extend $f$ on $K_0$ to a discrete Morse function $f$ on $K$ was presented in~\cite{KingKnudsonMramor2005}. 
In this section, we extend the work of King et al.~\cite{KingKnudsonMramor2005} to the setting of discrete stratified Morse theory. 
Let us first review the algorithm of~\cite{KingKnudsonMramor2005}. 
Since the function $f$ is injective, we may order the vertices. 
We begin with the vertex with smallest function value and proceed as follows. Given a vertex $v$, consider the lower link $K_v$ of $v$. If $K_v$ is empty then we know that $v$ is a local minimum and so we make $v$ critical. Otherwise, we restrict $f$ to $K_v$ and iteratively run the algorithm on $K_v$. During this iteration we take the extra step of canceling all possible gradient paths; that is, if there is a unique gradient path between two critical cells we reverse it to eliminate those critical cells. We then find the critical vertex $w$ in $K_v$ with smallest function value and pair $v$ with the edge $[v, w]$ (this makes sense as it should be the steepest edge away from $v$). For each regular pair $\sigma<\tau$ in $K_v$ we then pair $v\ast\sigma$ with $v\ast\tau$, and for each critical cell $\alpha\ne w$ in $K_v$ we make $v\ast\alpha$ critical. The resulting discrete vector field has no directed loops and is therefore a discrete gradient.  

To bring this into the stratified setting, we begin by assuming that we already have a stratification $\Scal=\{S_i\}$ of the complex $K$; let $s:K\to\Scal$ be the associated assignment map. Extend the partial order on $\Scal$ to a linear order if necessary and write the strata as $S_0<S_1<\cdots <S_n$. Given the function $f$ on $K_0$, consider the function $\text{maxf}$ on $K$ defined by setting $\text{maxf}(\sigma) = \max_{v\in\sigma} f(v)$. We then proceed as follows.

\begin{enumerate}\denselist
    \item The stratum $S_0$, being minimal in the order, is a subcomplex of $K$ (Lemma~\ref{lem:minimal-element}). Use the algorithm of ~\cite{KingKnudsonMramor2005} to generate a discrete Morse function $f_0$ on $S_0$ extending the restriction of $f$ to the vertices of $S_0$. We may choose such an extension to be arbitrarily close to the function $\text{maxf}$ (\cite{KingKnudsonMramor2005}, Theorem 3.4).
    
    \item Assume inductively that we have defined an extension $f_i$ on $S_i$, $i\ge 0$, that is a discrete stratified Morse function on $S_i$. The algorithm of~\cite{KingKnudsonMramor2005} works on $S_{i+1}\setminus S_i$ to generate a discrete Morse function on this space, with the following modification. Simplices adjacent to the boundary of $S_{i+1}$ may not be considered by the algorithm if the lower link of a vertex is empty. We therefore declare that all simplices that do not get considered remain unpaired (critical).
    \item In the end we obtain a discrete stratified Morse function $\tilde{f}:K\to\Rspace$ extending $f$.
\end{enumerate}

\begin{remark}
This algorithm leaves all simplices $\sigma$ having a face $\tau<\sigma$ with $s(\sigma)\ne s(\tau)$ critical. That is, the simplices in each stratum having a face in the stratum's frontier will be left unpaired by the algorithm.
\end{remark}

It is not clear that we can choose $\tilde{f}$ to be arbitrarily close to $\text{maxf}$ on all of $K$. Indeed, if the values of $f$ on lower strata are much larger than on higher strata it may not be possible to find such an extension in the inductive step. Moreover, in the inductive step, it could happen that a vertex in $S_{i+1}$ has an empty lower link, either because all its neighbors lie in $S_{i+1}$ and have higher values or because some of its neighbors lie in a lower stratum and are therefore not considered by the algorithm of \cite{KingKnudsonMramor2005}. This will force the vertex to be critical and in the latter case the adjacent simplices will be made critical as well, therefore making it impossible to keep associated function values close to the function $\text{maxf}$.

We do have the following curious result, however.

\begin{theorem}\label{thm:dsmfextdmf}
We may choose an extension $\tilde{f}:K\to\Rspace$ of $f$ that is a discrete Morse function on all of $K$.
\end{theorem}

\begin{proof}
The algorithm of \cite{KingKnudsonMramor2005} actually generates a discrete gradient vector field from the function $f:K_0\to\Rspace$. There is then a great deal of flexibility in choosing an extension $\tilde{f}$. Observe the following: in the inductive step we actually first generate a discrete gradient on $S_{i+1}\setminus S_i$ which happens to leave cells on the boundary critical (i.e., those simplices having a face in $S_i$ remain unpaired). We know that the union of these gradients is a discrete gradient on all of $K$ (Theorem~\ref{thm:dvf-union}) and we may then choose a discrete Morse function $\tilde{f}$ extending $f:K_0\to\Rspace$ compatible with this gradient.
%\qed
\end{proof}

Now, if we are not given a stratification of $K$, there are several ways we could proceed. We can choose some extension of $f$ to $K$, such as the function $\text{maxf}$ or the piecewise linear extension of $f$ (take the average value of the vertices of a simplex). Employing the algorithm of Section \ref{subsec:dsmf-algorithm} yields a stratification on which the extension is a discrete stratified Morse function. We could stop there, or we could discard the chosen extension and implement the algorithm above. Another approach is to use the algorithm of \cite{Nanda2019} to produce the coarsest stratification of $K$ into cohomology manifolds and then proceed using the algorithm above. It is not clear which method is preferable; this will be the subject of future research.

\section{Future Work}
\label{sec:future}

While we have laid the foundations for the study of a discrete version of stratified Morse theory and provided some examples and basic results, much work remains to be done. The most obvious question is:  can we formulate and prove discrete versions of the main theorems of stratified Morse theory, Theorem \ref{theorem:smt-a} and Theorem \ref{theorem:smt-b}? These relate the topology of sublevel sets in a stratified space to the critical points of a stratified Morse function.

There are a number of technical difficulties to be overcome to make this work. One hurdle is that discrete stratified Morse functions are not continuous in any reasonable sense. The function values may jump wildly from stratum to stratum and so examining the sublevel sets (or sublevel complexes) around a critical point in a lower-dimensional stratum may be problematic. So to have any hope, we will have to limit the types of functions we consider, such as insisting that the discrete function be a reasonable approximation to a smooth one.

There is also the matter of trying to understand the discrete analogue of the normal Morse data (see Appendix \ref{sec:prelim-SMT}). Since we are working with arbitrary simplicial complexes, it is not at all clear what the proper notion of ``normal slice" is, and so we must seek alternative formulations (see Theorem \ref{theorem:smt-c}). All of this is the focus of current research and the results will be presented elsewhere. 

\section*{Acknowledgments}
\label{sec:ack}

Bei Wang is supported in part by NSF IIS-1513616 and NSF ABI-1661375. 
We would like to thank Vin de Silva and Davide Lofano for their valuable comments regarding the definition of stratified simplicial complexes. 
%------------------------------------------------------------------------------------------------

%\bibliography{../dsmt}

\appendix
\newpage
\section{Preliminaries on classical and stratified Morse theory}
\label{sec:prelim-SMT}

For completeness, we include here a review of the basics of (stratified) Morse theory.
Given a topological space $\Xspace$, studying the relation between the critical points of a Morse function (or a stratified Morse function) on $\Xspace$ and the topology of $\Xspace$ requires more care in the smooth setting in comparison with the discrete setting. 
Most of our review originates from the seminal work of Goresky and MacPherson~\cite{GoreskyMacPherson1988}. 

\subsection{Classical Morse theory}
Let $\Xspace$ be a compact, differentiable $d$-manifold and $f: \Xspace \to \Rspace$ a smooth real-valued function on $\Xspace$.
For a given value $a \in \Rspace$, let $\Xspace_a = f^{-1}(-\infty,a] = \{x \in \Xspace \mid f(x) \leq a\}$ denote the \emph{sublevel set}.  
Morse theory studies the topological changes in $\Xspace_a$ as $a$ varies. 

\para{Morse functions.}
A point $x \in \Xspace$ is \emph{critical} if the derivative at $x$ equals zero. 
The value of $f$ at a critical point is a \emph{critical value}.
All other points are \emph{regular points} and all other values are \emph{regular values} of $f$. 
A critical point $x$ is \emph{non-degenerate} if the Hessian, the matrix of second partial derivatives at the point, 
is invertible. 
The \emph{Morse index} of the non-degenerate critical point $x$ is the number of negative eigenvalues in the Hessian matrix, denoted as $\lambda(x)$. 
\begin{definition}
\label{def:mf}
$f: \Xspace \to \Rspace$ is a \emph{Morse function} if all critical points are non-degenerate and its values at the critical points are distinct. 
\end{definition}

\para{Results.}
We now review two fundamental results of classical Morse theory (CMT). 

\begin{theorem}[CMT Theorem A] {\em (}\cite{GoreskyMacPherson1988}, p.~4; \cite{EdelsbrunnerHarer2010}, p.~129{\em )}
\label{theorem:cmt-a}
Let $f: \Xspace \to \Rspace$ be a differentiable function on a compact smooth manifold $\Xspace$. 
Let $a< b$ be real numbers such that $f^{-1}[a,b]$ is compact and contains no critical points of $f$. 
Then $\Xspace_a$ is diffeomorphic to $\Xspace_b$. 
\end{theorem}

On the other hand, let $f$ be a Morse function on $\Xspace$. 
We consider two regular values $a < b$ such that $f^{-1}[a,b]$ is compact but contains one critical point $u$ of $f$, 
with index $\lambda$. 
Then $\Xspace_b$ has the homotopy type of $\Xspace_a$ with a $\lambda$-cell (or $\lambda$-handle, the smooth analogue of a $\lambda$-cell) attached along its boundary
(\cite{GoreskyMacPherson1988}, page 5; \cite{EdelsbrunnerHarer2010}, page 129). 
We define \emph{Morse data} for $f$ at a critical point $u$ in $\Xspace$ to be a pair of topological spaces $(A,B)$ where $B \subset A$ with the property that as a real value $c$ increases from $a$ to $b$ (by crossing the critical value $f(u)$), the change in $\Xspace_c$ can be described by gluing in $A$ along $B$ \cite{GoreskyMacPherson1988} (page 4). 
Morse data measures the topological change in $\Xspace_c$ as $c$ crosses critical value $f(u)$. 
We have the second fundamental result of Morse theory,

\begin{theorem}[CMT Theorem B] {\em (}\cite{GoreskyMacPherson1988}, p.~5;
\cite{Matsumoto1997}, p.~77{\em )}
\label{theorem:cmt-b}
Let $f$ be a Morse function on $\Xspace$. 
Consider two regular values $a < b$ where $f^{-1}[a,b]$ is compact and contains one critical point $u$ of $f$, 
with index $\lambda$. 
Then $\Xspace_b$ is diffeomorphic to the space $\Xspace_a \cup_{B} A$, where 
$(A,B) = (D^{\lambda} \times D^{d-\lambda}, (\bdr D^{\lambda}) \times D^{d-\lambda})$ is the Morse data, $d$ is the dimension of $\Xspace$, $\lambda$ is the Morse index of $u$,
$D^{k}$ denotes the closed $k$-dimensional disk, and $\bdr D^{k}$ is its boundary.  
\end{theorem}

\subsection{Stratified Morse Theory}
\label{sec:smt}
Morse theory can be generalized to certain singular spaces, in particular to Whitney stratified spaces~\cite{GoreskyMacPherson1988,McTague2005}. 

\para{Stratified Morse function.}
Let $\Xspace$ be a compact $d$-dimensional Whitney stratified space embedded in some smooth manifold $\Mspace$. 
A function on $\Xspace$ is \emph{smooth} if it is the restriction to $\Xspace$ of a smooth function on $\Mspace$. 
Let $\bar{f}: \Mspace \to \Rspace$ be a smooth function.
The restriction $f$ of $\bar{f}$ to $\Xspace$ is \emph{critical} at a point $x \in \Xspace$ iff it is critical when restricted to the particular manifold piece which contains $x$ \cite{Bendich2008}.
A \emph{critical value} of $f$ is its value at a critical point.  
\begin{definition}
\label{def:smf}
$f$ is a \emph{{stratified Morse function}} if (\cite{Bendich2008}, \cite{GoreskyMacPherson1988} page 13): 
\begin{itemize}%\denselist
\item[1.] All critical values of $f$ are distinct.
\item[2.] At each critical point $u$ of $f$, the restriction of $f$ to the stratum $S$ containing $u$ is non-degenerate.  
\item[3.] The differential of $f$ at a critical point $u \in S$ does not annihilate (destroy) any limit of tangent spaces 
to any stratum $S'$ other than the stratum $S$ containing $u$.
\end{itemize}
\end{definition}
Condition 1 and 2 imply that $f$ is a Morse function when restricted to each stratum in the classical sense. 
Condition 2 is a non-degeneracy requirement in the tangential directions to $S$.
Condition 3 is a non-degeneracy requirement in the directions normal to $S$ \cite{GoreskyMacPherson1988} (page 13). 

\para{Results.}
Now we state the two fundamental results of stratified Morse theory. 

\begin{theorem}[SMT Theorem A] {\em (}\cite{GoreskyMacPherson1988}, p.~6{\em )}
\label{theorem:smt-a}
Let $\Xspace$ be a Whitney stratified space and $f: \Xspace \to \Rspace$ a stratified Morse function. 
Suppose the interval $[a,b]$ contains no critical values of $f$. 
Then $\Xspace_a$ is diffeomorphic to $\Xspace_b$. 
\end{theorem}

\begin{theorem}[SMT Theorem B] {\em (}\cite{GoreskyMacPherson1988}, p.~8 and p.~64{\em )}
\label{theorem:smt-b}
Let $f$ be a stratified Morse function on a compact Whitney stratified space $\Xspace$.
Consider two regular values $a < b$ such that $f^{-1}[a,b]$ is compact but contains one critical point $u$ of $f$. 
Then $\Xspace_b$ is diffeomorphic to the space $\Xspace_a \cup_B A$, where the Morse data $(A,B)$ is the product of the normal Morse data at $u$ and the tangential Morse data at $u$. 
\end{theorem}

To define tangential and normal Morse data, we have the following setup.
Let $\Xspace$ be a Whitney stratified subset of some smooth manifold $\Mspace$.
Let $f: \Xspace \to \Rspace$ be a stratified Morse function with a critical point $u$. 
Let $S$ denote the stratum of $\Xspace$ which contains the critical point $u$.
Let $N$ be a \emph{normal slice} at $u$,  that is, 
$N = \Xspace \cap N' \cap B^{\Mspace}_{\delta}(u)$, 
where $N'$ is a sub-manifold of $\Mspace$ which is traverse to each stratum of $\Xspace$, 
intersects the stratum $S$ in a single point $u$, and satisfies $\dime{S} + \dime{N'} = \dime{\Mspace}$.  
$B_\delta^{\Mspace}(u)$ is a closed ball of radius $\delta$ in $\Mspace$ based on a Riemannian metric on $\Mspace$. 
By Whitney's condition, if $\delta$ is sufficiently small then $\bdr B^{\Mspace}_{\delta}(u)$ will be transverse to each stratum of $\Xspace$, and to each stratum in $\Xspace \cap N'$,  fix such a $\delta > 0$ \cite{GoreskyMacPherson1988} (page 40). 

The \emph{tangential Morse data} for $f$ at $u$ is the pair 
$$(P, Q) = (D^{\lambda} \times D^{s-\lambda}, (\bdr D^{\lambda}) \times D^{s - \lambda}),$$
where $\lambda$ is the (classical) Morse index of $f$ restricted to $S$, $f|S$, at $u$, 
and $s$ is the dimensional of stratum $S$ \cite{GoreskyMacPherson1988} (page 65). 

The \emph{normal Morse data} is the pair
$$(J, K) = (N \cap f^{-1}[v-\ep, v+\ep], N \cap f^{-1}(v-\ep)),$$
where $f(u) = v$ and $\ep > 0$ is chosen such that $f | N$ has no critical values other than $v$ in 
the interval $[v - \ep, v+\ep]$ \cite{GoreskyMacPherson1988} (page 65).

The \emph{Morse data} is the topological product of the tangential and the normal Morse data, 
where the product of pairs is defined as $(A, B) = (P, Q) \times (J,K) = 
(P \times J , P \times K \cup Q \times J)$. 

Theorem~\ref{theorem:smt-b} corresponds to the Main theorem of~\cite{GoreskyMacPherson1988} (page 65), which  has the following homotopy consequences.
Suppose $\Xspace$ is a Whitney stratified space, $f: \Xspace \to \Rspace$ is a proper stratified Morse function, and $[a,b]$ contains no critical values except for a single isolated critical value $v \in (a,b)$ which corresponds to a critic point $p$ in some stratum $\Sspace$ of $\Xspace$. 
$\lambda$ is the Morse index of $f|_{\Sspace}$ at the point $p$.

\begin{theorem}[SMT Homotopy Consequences] {\em (}\cite{GoreskyMacPherson1988} (Section 3.12, p.~68{\em )}
\label{theorem:smt-c}
The space $\Xspace_{b}$ has the homotopy type of a space which is obtained from $\Xspace_{a}$ by attaching the pair 
$$(D^{\lambda}, \bdr D^{\lambda}) \times (\cone(l^{-}), l^{-}).$$
\end{theorem}
Here, $l^{-}$ is the \emph{lower half link} of $\Xspace$ where $l^{-} = N \cap f^{-1}(v-\ep) \cap B_{\delta}^{\Mspace}$.

%------------------------------------------------------------------------------------------------
\end{document}